\documentclass[aps, pra, reprint, superscriptaddress]{revtex4-2}
\usepackage[utf8]{inputenc}

\usepackage{amsmath, amssymb, amsthm, bm, physics}
\usepackage{xcolor, graphicx, subfigure, float, hyperref}
\usepackage{tcolorbox}
\usepackage{dsfont}
\usepackage{booktabs}
\usepackage{microtype}
\newtheorem{thm}{Theorem}[]
\newtheorem{lem}{Lemma}
\newtheorem{prop}{Proposition}

\usepackage{hyperref}
\definecolor{forestgreen}{rgb}{0.13, 0.55, 0.13}

\begin{document}
\title{Semi-Device Independent Non-stabilizerness Certification in the Prepare-and-Measure Scenario}

\author{Santiago Zamora}
\affiliation{Physics Department, Federal University of Rio Grande do Norte, Natal, 59072-970, Rio Grande do Norte, Brazil}
\affiliation{International Institute of Physics, Federal University of Rio Grande do Norte, 59078-970, Natal, Brazil}
\author{Rafael A. Macêdo}
\affiliation{Physics Department, Federal University of Rio Grande do Norte, Natal, 59072-970, Rio Grande do Norte, Brazil}
\affiliation{International Institute of Physics, Federal University of Rio Grande do Norte, 59078-970, Natal, Brazil}
\author{Tailan S. Sarubi}
\affiliation{Physics Department, Federal University of Rio Grande do Norte, Natal, 59072-970, Rio Grande do Norte, Brazil}
\affiliation{International Institute of Physics, Federal University of Rio Grande do Norte, 59078-970, Natal, Brazil}
\author{Moisés Alves}
\affiliation{Physics Department, Federal University of Rio Grande do Norte, Natal, 59072-970, Rio Grande do Norte, Brazil}
\affiliation{International Institute of Physics, Federal University of Rio Grande do Norte, 59078-970, Natal, Brazil}
\author{Davide Poderini}
\affiliation{Universit\`a degli Studi di Pavia, Dipartimento di Fisica, QUIT Group, via Bassi 6, 27100 Pavia, Italy}
\affiliation{International Institute of Physics, Federal University of Rio Grande do Norte, 59078-970, Natal, Brazil}
\author{Rafael Chaves}
\affiliation{International Institute of Physics, Federal University of Rio Grande do Norte, 59078-970, Natal, Brazil}
\affiliation{School of Science and Technology, Federal University of Rio Grande do Norte, Natal, Brazil}

\begin{abstract}
Non-stabilizerness is an essential resource for quantum computational advantage, as stabilizer states admit efficient classical simulation. We develop a semi-device-independent framework for certifying non-stabilizer states in prepare-and-measure (PAM) scenarios, relying only on assumptions about the system's dimension. Within this framework, we introduce prepare-and-measure witnesses that can distinguish stabilizer from non-stabilizer states, and we provide analytical proofs that threshold violations of these witnesses certify non-stabilizerness. In the simplest setting—three preparations, two measurements, and qubit systems—surpassing a specific threshold guarantees that at least one prepared state lies outside the stabilizer polytope, while a stronger violation can certify at least two. We extend this approach by linking it to quantum random access codes, also generalizing our results to qutrit systems and introducing a necessary condition for certifying non-stabilizerness based on state overlaps (Gram matrices). These results offer a set of semi-device-independent tools for practically and systematically verifying non-stabilizer states using prepare-and-measure inequalities.
\end{abstract}

\maketitle

\section{Introduction}

Quantum computation offers exponential speed-ups over classical algorithms, yet not all quantum resources are equally powerful. The Gottesman-Knill theorem \cite{gottesman1998heisenbergrepresentationquantumcomputers} illustrates this by showing that stabilizer states, although entangled, are classically simulable \cite{PhysRevA.70.052328}. Consequently, non-stabilizer states—those lying outside the stabilizer formalism—are essential for achieving universal fault-tolerant quantum computation \cite{bravyi2005universal, howard2017application}. This resource, known as non-stabilizerness \cite{bravyi2005universal}, underpins quantum advantage \cite{gottesman1997stabilizer}.

Non-stabilizerness arises because Clifford operations and stabilizer states alone form a subtheory of quantum mechanics that admits efficient classical simulation \cite{bravyi2005universal}. Non-stabilizer states, such as the T-state or the qutrit strange~ state \cite{Heinrich2019robustnessofmagic, Veitch_2014}, when used as ancillary resources, elevate this subtheory to full universality via state injection and distillation protocols \cite{Souza2011,10.1007/978-3-642-10698-9_3}. Non-stabilizer operations, which cannot be built from Clifford gates, also possess non-stabilizerness and often present implementation bottlenecks in practice.

Recent advances in resource theories have formalized non-stabilizerness through quantifiers such as robustness of magic \cite{PhysRevLett.118.090501, Heinrich2019robustnessofmagic, Hamaguchi2024handbookquantifying}, mana \cite{Veitch_2014}, trace distance \cite{junior2025geometricanalysisstabilizerpolytope}, and stabilizer entropies \cite{PhysRevLett.128.050402}, each capturing distinct facets of non-stabilizerness. These measures admit operational interpretations and enable rigorous analysis of how non-stabilizerness evolves under quantum processes. For instance, robustness of non-stabilizerness quantifies the minimum mixture with stabilizer states required to represent a given state, while mana links to the negativity of the discrete Wigner function \cite{Veitch_2014, 10.1063/1.2393152, PhysRevA.73.012301}, reflecting its quasi-probabilistic character.

Geometrically, stabilizer states form convex polytopes in the state space whose complexity scales exponentially with system size \cite{Heinrich2019robustnessofmagic}. Though this complexity hampers direct analysis in large systems, studies of few-qudit systems have revealed important connections between geometry, entanglement, and non-stabilizerness \cite{reichardt2009quantumuniversalitystatedistillation, junior2025geometricanalysisstabilizerpolytope, warmuz2024magicmonotonefaithfuldetection}.

Given the central role of non-stabilizerness, efficient certification methods are critical. Inspired by the use of Bell inequalities to detect entanglement, one might consider device-independent (DI) techniques to certify non-stabilizerness. However, such approaches face obstacles due to the dimension-sensitive nature of stabilizer states. For instance, the entangled state described by the superposition $\frac{1}{\sqrt{2}}(\ket{00} + \ket{11})$ is a stabilizer state for a two-qubit ($d=2$) system. Yet, this same mathematical state is a non-stabilizer state when considered within the larger Hilbert space of a two-qutrit ($d=3$) system. This dependence of a state's stabilizer status on the underlying dimension fundamentally complicates any purely device-independent certification method.

Nonetheless, specially tailored Bell inequalities have been shown to detect the presence of non-stabilizerness \cite{macedo2025witnessingmagicbellinequalities}. These methods, however, depend on entanglement and are invariant under local unitaries, rendering them ineffective for certifying non-stabilizerness in single systems. This motivates semi-device-independent (SDI) approaches, which impose operational constraints without full device characterization. In particular, the prepare-and-measure (PAM) scenario \cite{PhysRevLett.105.230501} emerges as a natural SDI framework for certifying non-stabilizerness without the need for multipartite systems.

The PAM scenario plays a versatile role in quantum information \cite{tavakoli2025unlimitedquantumcorrelationadvantage, svegborn2025quantuminputsprepareandmeasurescenario, de2021general}. Here, a preparation device encodes data into classical or quantum states of fixed dimension, which are then sent to a measurement device. PAM implementations rely solely on dimensional assumptions, although alternative constraints such as entropy \cite{PhysRevLett.115.110501}, distinguishability \cite{brask2017megahertz}, or photon number \cite{pauwels2024prepare} may be considered. PAM scenarios have found applications in quantum communication, QRACs \cite{PhysRevA.85.052308, ambainis2009quantumrandomaccesscodes}, randomness generation \cite{PhysRevA.100.062338, pawlowski2011semi}, entanglement certification \cite{moreno2021semi, tavakoli2021correlations, pauwels2022entanglement}, self-testing \cite{tavakoli2020self}, and foundational tests \cite{pawlowski2009information,chaves2015information,PhysRevLett.134.120203, Divianszky2023, chaves2018causal, polino2019device}. Their flexibility allows robust certification of quantum advantages using prepare-and-measure inequalities (PI) \cite{PhysRevLett.105.230501}.

Our work builds upon recent progress in non-stabilizerness certification. The concept of ``set-magic" in an ensemble of states was introduced in Ref.~\cite{wagner2024certifyingnonstabilizernessquantumprocessors}, which proposed a certification method based on linear overlap inequalities. Concurrently, a detailed geometric (device-dependent) analysis of the stabilizer polytope, along with a  useful resource quantifier, was developed in Ref.~\cite{junior2025geometricanalysisstabilizerpolytope}. The present work bridges these ideas: we introduce a novel, semi-device-independent framework to certify the set-magic of Ref.~\cite{wagner2024certifyingnonstabilizernessquantumprocessors}, and we use the geometric quantifier from Ref.~\cite{junior2025geometricanalysisstabilizerpolytope} to analyze our results. Our primary contribution is thus an operational, semi-device-independent and experimentally friendly certification method within the prepare-and-measure scenario.

In this work, we develop witnesses based on prepare-and-measure inequalities that, in other contexts, have been used as dimension witnesses~\cite{PhysRevLett.105.230501}, and on average success probabilities of QRAC's~\cite{ambainis2009quantumrandomaccesscodes,Farkas2019}.     Here, we repurpose them for the distinct task of certifying the non-stabilizerness resource within a fixed-dimension scenario. By constraining the dimension of transmitted messages, we derive criteria distinguishing stabilizer from non-stabilizer preparations. Extending the approach of \cite{HORODECKI1995340}, we generalize the quantum violation criterion introduced in \cite{PhysRevResearch.2.043106} to a broader class of PI.

This paper is organized as follows. Section \ref{Sec: stabilizer_formalism} reviews the stabilizer formalism, and Section \ref{Sec: PAM_scenario} formalizes the prepare-and-measure scenario. In Section \ref{sec:PAM_simplest_scenario}, we analyze the minimal PAM setting, showing  explicit witnesses for non-stabilizerness and demonstrating the self-testing of H-type magic states. We then extend this analysis in Section \ref{sEC:QUANTUM RANDOM ACCESS CODES} by linking our framework to quantum random access codes (QRACs), which allow for a finer-grained certification and a generalization to qudit systems. In Section \ref{sec:Separating Witnesses}, we introduce a necessary geometric condition for certification based on state overlaps (Gram matrices). We conclude in Section \ref{sec:FINAL REMARKS} with final remarks and perspectives.

\section{Stabilizer formalism for qudits}\label{Sec: stabilizer_formalism}
In this work, we focus on the detection of non-stabilizer states in qudit systems. To this end, we begin by introducing and characterizing the set of all stabilizer states, which form a convex structure known as the stabilizer polytope.

Let us begin by defining a pair of operators that act on qudits states described in a Hilbert space $\mathcal{H} \cong \mathbb C^d$. In particular, given the computational basis $\{\ket{j}\}_{j \in \mathbb{F}_d}$ with $\mathbb{F}_d = \{0, 1, \dots, d-1\}$, where $d$ is a prime number, one can define the Shift and Boost operators \cite{Veitch_2014,macedo2025witnessingmagicbellinequalities}:
\begin{equation}
    X\ket{j} = \ket{j+1 \, \text{mod } d}, \hspace{1cm} Z\ket{j} = \omega^j \ket{j},\label{Boost_and_shift}
\end{equation}
where $\omega = \exp(2\pi i/d)$. These operators satisfy $X^d=Z^d=I$ and $Z^{a}X^b = \omega^{ab} X^bZ^{a}$ and they generate the qudit Pauli group defined as  
\begin{equation}
    \mathcal{P}_{d} = \{D_{a, b} = \omega_d^{ab} X^{a} Z^{b} \mid (a,b) \in \mathbb{F}_{d} \times \mathbb{F}_{d} \}, \label{eq: Weyl_Heisenber_group}
\end{equation}
where \( \omega_d = \omega^{2^{-1}} \) if \( d > 2 \), and \( \omega_d = i \) if \( d = 2 \). All operations are taken modulo \( d \); in particular, \( 2^{-1} \) denotes the multiplicative inverse of \( 2 \) in \( \mathbb{F}_d \).  
The elements $D_{a ,b}$ are called displacement operators and the group multiplication is defined by
\begin{equation}
    D_{a,b}\cdot D_{a'b'} = \omega_d^{a'b-ab'}D_{a+a',b+b'} \label{eq: WHgroup_operation}.
\end{equation}

The Clifford group \( \mathcal{C}\ell_d \) is defined as the normalizer of \( \mathcal{P}_d \) within the unitary group \( \mathcal{U}(d) \). That is, it consists of all unitary operators that preserve \( \mathcal{P}_d \) under conjugation:
\begin{equation}
    \mathcal{C}\ell_d = \{ C \in \mathcal{U}(d) \mid C P C^\dagger \in \mathcal{P}_d \;,\;\forall P \in \mathcal{P}_d \}.
\end{equation}
In other words, Clifford unitaries map Pauli operators to other Paulis, preserving the group structure. Given a Clifford unitary $C \in \mathcal C\ell_d$, there is a state $C|0\rangle$ which must be \emph{stabilized} by the Pauli $P = CZC^\dagger$, that is, $P \;C|0\rangle = + C|0\rangle$. We call such a state a stabilizer state. The group generated by $P$, denoted as $\mathcal S = \langle P \rangle$, is called the stabilizer group. We will use $|\mathcal S\rangle $ to refer to the stabilizer state stabilized by $\mathcal S$.

In dimension $d$, there are $d(d+1)$ stabilizer states \cite{Veitch_2014}, whose convex hull defines the stabilizer polytope:
\begin{equation}
    \mathrm{STAB}_{d} = \mathrm{conv}\{ |\mathcal S\rangle \langle \mathcal S| \mid \mathcal S \mathrm{\;is\;a\;stabilizer\;group} \} \subseteq \mathcal D(\mathbb C^{d})\;, \label{eq: StabPolytope_Vrep}
\end{equation}
where $\mathcal D(\mathbb C^d)$ refers to the space of density matrices on $\mathbb C^d$. That is, if $\rho \in \mathrm{STAB}_d$, $\rho = \sum_{\mathcal S} p_\mathcal S |\mathcal S\rangle \langle \mathcal S|$, with $p_\mathcal S\geq 0$ and $\sum_\mathcal Sp_\mathcal S=1$. This is the so-called $V$-representation of a polytope \cite{grunbaum1967convex}, where the vertices are specified. Alternatively, the stabilizer polytope can be characterized via its \( H \)-representation \cite{junior2025geometricanalysisstabilizerpolytope}:
\begin{equation}
    \mathrm{STAB}_d = \left\{ \rho \in \mathcal{D}(\mathbb C^d) \,\middle|\, \mathrm{Tr}(\rho A^{\vec {q}}) \geq 0,\; \vec{q} \in \mathbb{F}_d^{d+1} \right\}, \label{eq: H-rep}
\end{equation}
where the operators \( A^{\vec{q}} = -I_d + \sum_{j=1}^{d+1} \Pi_j^{q_j} \) define the facets of the polytope \cite{Howard2014}. Here, \( \Pi_j^{q_j} \) denotes the rank-$1$ projector on the eigenvector corresponding to the eigenvalue \( \omega^{q_j} \) of the \( j \)-th operator in the set \( \{D_{0,1}, D_{1,0}, D_{1,1},\ldots, D_{1,d-1}\} \subset \mathcal{P}_d \).

The stabilizer polytope partitions the state space into two disjoint subsets: the polytope itself, which contains all stabilizer states, and its complement, non-stabilizer states $\rho \notin \mathrm{STAB}_d$, also referred to as \emph{magic} states. Based on their importance, a quantum resource theory \cite{chitambar2019quantum} of non-stabilizerness can also be formulated \cite{Veitch_2014}, where the sets of free states are $\mathrm{STAB}_d$, and free operations are Clifford unitaries and Pauli measurements~\footnote{This assumes a fixed dimension $d$, corresponding to what is known as the \emph{stabilizer subtheory}. If the dimension is not fixed, then one can also consider tensor products with stabilizer ancillae qudits and partial traces as free operations}.

%\rafa{I will move something like this to the introduction: ``They are of profound importance to quantum computing, since they are the necessary resource for quantum computers to perform universal quantum computation in a fault-tolerant manner \cite{bravyi2005universal,howard2017application}''}

On the other side, in quantum networks, cryptographic protocols, or prepare-and-measure scenarios, one often deals with collections of independently prepared states. The operational power in such settings often derives from how distinguishable, incompatible, or non-orthogonal these states are with respect to a free set. However, in some cases, the certification of non-classicality in these scenarios can also depend on a relational aspect between the 
elements of the set of independent states. As we will show in this paper, this is indeed the case when the task is to certify non-stabilizerness in the PAM scenario in a SDI independent way. We will consider the so-called \textit{set-magic} first defined in Ref.~\cite{wagner2024certifyingnonstabilizernessquantumprocessors}. 
It characterizes, from a relational point of view, the non-stabilizerness of a finite set of states $\underline \rho = \{\rho_i\}_i \subset \mathcal{D}(\mathbb C^d)$. Its definition is as follows: consider an ensemble $\underline \rho \subset \mathcal{D}(\mathbb C^d)$ consisting of a finite number of elements, each of dimension $d$. We say that $\underline \rho$ has set-magic if there is no unitary $U$ such that $U \underline \rho U^\dagger \subseteq \mathrm{STAB}_d$. In other words, at least one element in the set must be outside the stabilizer set.

There are various proposals for the quantum verification of stabilizerness, and in particular, there are efficient protocols for non-stabilizerness quantification \cite{PRXQuantum.4.010301,PhysRevLett.132.240602,haug2025efficientwitnessingtestingmagic}. However, these approaches are resource-intensive and rely on several assumptions and requirements.  In the next section, we introduce and discuss PAM scenarios, which are semi-device-independent (SDI) protocols in the sense that they require only a single assumption: the quantum dimension $d$ of the state. This minimal assumption makes PAM witnesses natural candidates for detecting set-magic.

\section{PAM scenario and a hierarchy of three sets of correlations.} \label{Sec: PAM_scenario}

The PAM scenario involves two main devices, as seen in Fig. \ref{fig: PAM_scenario}: a preparation device that generates physical systems based on external inputs $x \in \mathbb X =  \{1,\cdots, |\mathbb X|\}$ described by a random variable $X$, and a measurement device that performs measurements conditional on another external random variable  $Y$ with possible values $y \in \mathbb Y =  \{1,\cdots, |\mathbb Y|\}$. The measurement outputs $b \in \mathbb B=  \{1,\cdots, |\mathbb B|\}$ are represented by a third random variable $B$. This framework allows for a rigorous analytical and practical approach for examining fundamental distinctions between classical and quantum physical theories, particularly concerning dimensional and causal constraints imposed on the systems under study \cite{PhysRevLett.105.230501}. 

In the quantum description of the PAM scenario, given some fixed dimension $d$, the random variable $X$ chooses $\rho_{x \in \mathbb X} \in \mathcal D(\mathbb C^d)$, while $y \in \mathbb Y$ fixes the measurement $M_b^{y}$ to be performed with possible outcome $b \in \mathbb B$ . By applying Born's rule, the observed measurement probability is given by
\begin{equation}
    p(b|x,y) = \text{Tr}(\rho_x M_b^y),\label{eq: born_rule}
\end{equation} 
where $\sum_{b \in \mathbb B}M_{b}^y =1$ and $M_{b}^y\succeq0  \quad\forall y,b$. 

If the distribution admits a classical explanation,
\textit{i.e.} if $p(b|x,y)$ is such that it can be decomposed as
\begin{equation}
    p(b|x,y) = \sum_{\lambda, m} p(b|m,y,\lambda) p(m|x,\lambda) p(\lambda),
    \label{eq:classicalprob}
\end{equation}
where $\lambda$ is a latent variable that allows for correlations between the preparation and measurement devices, then the quantum message can be replaced instead by a classical random variable $M$ taking values $m\in \mathbb{M} =\{1,\dots,|\mathbb{M}|\} $ with  $|\mathbb{M}|=d$.  In Fig. \ref{fig: PAM_scenario} we show the associated causal structure represented by a directed acyclic graph (DAG). 

 Given the tuple $(d, \mathbb X,\mathbb Y,\mathbb B)$, we say that a behaviour $\{p(b|x,y)\}_{b \in \mathbb B, x \in \mathbb X, y \in \mathbb Y}$ is in $\mathcal C_{C,d} \subseteq [0,1]^{|\mathbb B||\mathbb X||\mathbb Y|}$ if it admits the classical explanation written in  Eq.~(\ref{eq:classicalprob}), or in $\mathcal C_{Q,d}\subseteq [0,1]^{|\mathbb B||\mathbb X||\mathbb Y|}$ if it admits a quantum state-measurement decomposition as in Eq.~(\ref{eq: born_rule}). The classical set is known to be a polytope \cite{PhysRevLett.105.230501}, where the extrema are given by deterministic strategies. 

\begin{figure}[t!]
    \centering
    \includegraphics[width=1.0\linewidth]{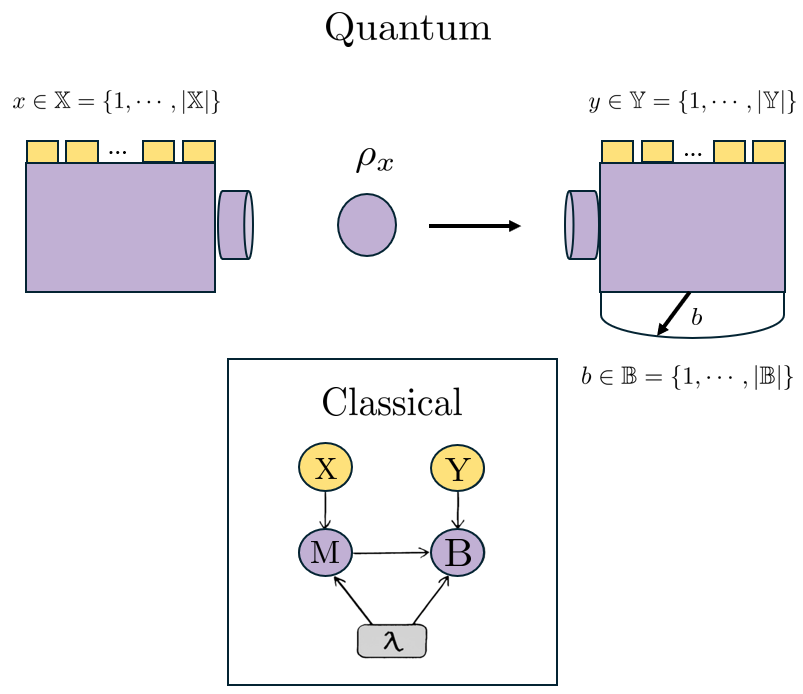}
    \caption{Illustration of the PAM scenario. A  preparation device selects an input $X = x \in \{1, \cdots, |\mathbb X|\}$, producing a quantum state represented by the density matrix \( \rho_x \). This state is subsequently measured by a measurement device configured with input $Y = y \in \{1, \cdots, |\mathbb Y|\}$, yielding outcome $B =b\in\{1, \cdots, |\mathbb B|\}$. Inside the square, we show the associated causal structure when considering a classical hidden variable $\lambda$ between the preparation and measurement device. }
    \label{fig: PAM_scenario}
\end{figure}
 
Now we introduce a third set $\mathcal C_{\mathrm{STAB},d} \subseteq [0,1]^{|\mathbb B||\mathbb X||\mathbb Y|}$ composed of  correlations generated only by stabilizer states:
\begin{equation}
    p(b|x,y) = \mathrm{Tr}(\sigma_x M^y_b), \quad \quad \sigma_x \in \mathrm{STAB}_d\;.
\end{equation}
 We now show some general relations between the 3 sets. 
\begin{prop}
    Let $\underline \rho = \{\rho_x\}_x$ be a set of states and $\underline M = \{M^y_b\}_{b,y}$ be a set of measurements. The map  $(\underline \rho, \underline M) \mapsto \{p(b|x,y) = \mathrm{Tr}(\rho_x M^{y}_b)\}_{b \in B, x \in X, y \in Y} \in \mathcal C_{Q,d}$ has the same image for $(\underline \rho, \underline M)$ and $(U\underline \rho U^\dagger, U \underline M U^\dagger)$ for any $U \in \mathcal U(d)$.
    \label{prop:unitary_symmetry}
\end{prop}
\begin{proof}
    This follows from the cyclic property of the trace:
\begin{equation}
    p(b|x,y) = \mathrm{Tr}(\rho_x M^{y}_b) = \mathrm{Tr}(U \rho_x U^\dagger U M^y_b U^\dagger)\;.
\end{equation}
\end{proof}
This symmetry results in strong constraints for the correlation sets:
\begin{lem}
 (1) $\mathcal C_{C,d} \subseteq \mathcal C_{\mathrm{STAB},d} \subseteq \mathcal C_{Q,d}$ and (2) if $|\mathbb X| \leq d$, we say that the scenario becomes trivial, meaning that $\mathcal C_{C,d} = \mathcal C_{\mathrm{STAB},d} = \mathcal C_{Q,d}$.
 \label{lem:non-trivial}
\end{lem}
\begin{proof}
    (1) The inclusion $ \mathcal C_{\mathrm{STAB},d} \subseteq \mathcal C_{Q,d}$ follows from the fact that stabilizer states are a subset of all quantum states, $\mathrm{STAB}_d\subseteq\mathcal{D(\mathcal{H}_d})$. The inclusion $\mathcal C_{C,d} \subseteq \mathcal C_{\mathrm{STAB},d}$ can be seen by noting that any classical probability distribution can be reproduced by using quantum states that are diagonal in a fixed basis; these states are a subset of the stabilizer polytope (up to a unitary rotation). 
    (2) The equality of the sets is a foundational result in the study of dimension witnesses\cite{PhysRevLett.105.230501, PhysRevResearch.2.043106}. The core idea is that when the number of preparations does not exceed the system's dimension, each input $x$ can be encoded into a perfectly distinguishable  quantum state (consider an orthogonal state basis for example). Since any arbitrary measurement statistics can be recovered by performing a suitable measurement on this set of distinguishable states, no quantum advantage is possible, and the correlation sets coincide.
\end{proof}
The usual way to characterize the observed PAM correlations is through prepare-and-measure inequalities, which have been proposed before as \emph{dimensional witnesses}~\cite{PhysRevLett.105.230501}. Given a scenario $S \in \{C, \mathrm{STAB},Q \}$ ($S$ can represent the classical, stabilizer or quantum scenario), the set of numbers $\{W^b_{x,y}\}_{b \in \mathbb B, x \in \mathbb X, y \in \mathbb Y} \subseteq \mathbb R$  defines the prepare-and-measure witness 
\begin{equation}
    W^S \geq W= \sum_{b,x,y} W^b_{x,y}p(b|x,y)
\end{equation}
where 
\begin{align}
    W^S &\equiv \max_{p \in \mathcal C_{S,d}} \sum_{b,x,y} W^b_{x,y}p(b|x,y) \nonumber\\ &=\max_{\substack{  \{\rho_x\} \subseteq S \\\{M^y\} \subseteq \mathcal M_\mathbb B(\mathbb C^d)}    } \sum_{b,x,y}W^{b}_{x,y} \mathrm{Tr}(\rho_x M^b_y)\;,
    \label{eq:state-measurement_opt}
\end{align}
denotes the $S$ value  of $W$.  $\mathcal M_\mathbb B(\mathbb C^d)$ denotes the space of $|\mathbb B|$-outcome measurements on $\mathbb C^d$. The hierarchy on Lemma \ref{lem:non-trivial} guarantees that $W^C \leq W^{\mathrm{STAB}} \leq W^Q$.

Now, given some set $\{\rho_x\}_{x \in X} \subseteq \mathcal D(\mathbb C^d)$ and a witness $W$,   considering all possible measurements defines the maximum value achievable by $W$ on this set of states as
\begin{equation}
    W(\{\rho_x\}_x) \equiv \max_{M^y \in \mathcal M_\mathbb B(\mathbb C^d)} \sum_{b,x,y} W^b_{x,y}\mathrm{Tr}(\rho_x M^y_b)\; .
    \label{eq:state_witness}
\end{equation}
Consider the violation $W(\{\rho_x\}_x)> W^S$. It   has the following interpretation.
\begin{lem}\label{lemma: violation_unitary}
    Let $W(\{\rho_x\}_x) > W^S$. Then, $\nexists U \in \mathcal U(d)$ such that $\{U \rho_x U^\dagger \}_x \subset S$.
\end{lem}
\begin{proof}
    Note that $W(\{U \rho_x U^\dagger\}_x) = W(\{\rho_x\}_x)$, since the optimization over measurements is invariant under unitaries. Then, if such a unitary exists, we have $W(\{\rho_x\}) = W(\{U \rho_x U^\dagger\}_x \subset S) \leq W^S$.
\end{proof}

Hence, a violation certifies the non-inclusion of the set $\{\rho_x\}_{x \in X}$ in $S$. The scenario $S=C$ arises when one  considers \textit{set-incoherent} ensembles of states \cite{PhysRevLett.126.220404,wagner2024certifyingnonstabilizernessquantumprocessors,PhysRevA.101.062110}. \textit{Set-coherence} was first introduced in \cite{PhysRevLett.126.220404}, and is a basis-independent definition of coherence which measures if a set of states can be written diagonally on the same basis. In the case of $S=\mathrm{STAB}$, it falls exactly into the set-magic definition of the last section.

In the next sections, we will provide explicit witnesses that show that in various PAM scenarios, the strict inclusion hierarchy holds,
\begin{equation}
    \mathcal C_{C,d} \subset \mathcal C_{\mathrm{STAB},d} \subset \mathcal C_{Q,d}\;.
    \label{eq:strict_hierarchy}
\end{equation}

\section{Simplest PAM scenario}\label{sec:PAM_simplest_scenario}
\subsection{The standard witness $S_3$}\label{section: A simple analytical witness}
In this section, we will show that the simplest PAM scenario already separates the classical-stabilizer-quantum correlations. The scenario is defined by the tuple $(d=2, |\mathbb X| = 3, |\mathbb Y| = 2, |\mathbb B|=2 )$, which is in fact the first non-trivial scenario allowed by Lemma \ref{lem:non-trivial}. Given $x \in \mathbb X$ and $y \in \mathbb Y$, define $E_{xy} =p(0|x,y)-p(1|x,y)$. The corresponding classical polytope $\mathcal C_{C,d}$ is fully described by the tight prepare-and-measure inequality \cite{PhysRevLett.105.230501}
\begin{equation}
    S_3 = E_{11} + E_{12} + E_{21} - E_{22}- E_{31} \leq S^C_3 =3, \label{eq: S3}
\end{equation}
and its $24$ symmetries corresponding to the relabeling of the preparation inputs and measurements inputs and outputs. It is to be compared to the maximal qubit value of $S^Q_3 = 1 +
2\sqrt{2} \approx 3.82$, which can be obtained from the following optimal preparations and measurements \cite{PhysRevLett.105.230501}
\begin{equation}
\rho_x = \frac{\mathbb{I} + \vec{r}_x \cdot \vec{\sigma}}{2}, \quad M_b^{y} = \frac{\mathbb{I} + (-1)^b \vec{s}_{y} \cdot \vec{\sigma}}{2} ,\label{eq: states_measurements}
\end{equation}
with 
\begin{equation}
\vec{s}_y = \frac{\vec{r}_1 + (-1)^{y+1}\vec{r}_2}{\sqrt{2}}, \quad
\vec{r}_3 = \frac{-\vec{r}_1 - \vec{r}_2}{\sqrt{2}}, \label{eq: optimalMS_S3}
\end{equation}
and where \(\vec{\sigma} = (X, Y, Z)\) denotes the vector of Pauli matrices. Remarkably, it turns out that $S_3$ also witnesses non-stabilizerness.
\begin{lem}
    \begin{equation}
        S^C_3 < S_3^\mathrm{STAB} = \sqrt 5 + \sqrt 2 < S^Q_3\;.
    \end{equation}
\end{lem}
\begin{proof}
    Consider the states $\rho_x = (1+ \vec r_x \cdot \vec \sigma)/2$ for $x \in \{1,2,3\}$. Ref.~\cite{PhysRevResearch.2.043106} showed  that optimization over measurements of $S_3$ yields
    \begin{eqnarray}
    \label{eq.S3withr}
    S_3(\{\vec r_x \}_x) = 
    \max_{\{\vec r_x \}_x}\|\vec{r}_1 + \vec{r}_2 - \vec{r}_3\|_2 + \|\vec{r}_1 - \vec{r}_2\|_2.
    \end{eqnarray}
    We consider $S =\mathrm{STAB}$, where $\vec r_x \in \{\pm e_1, \pm e_2, \pm e_3\}$, with $e_1 = (1,0,0)$, $e_2 = (0,1,0)$ and $e_3 = (0,0,1)$. Then, the inequality is optimized by choosing $\vec r_2 \perp \vec r_1$ and $\vec r_3 =-\vec r_1$, yielding $S^\mathrm{STAB}_3=\sqrt 5 + \sqrt 2$.
\end{proof}

Upon substituting the numerical values, we have $3 < S^\mathrm{STAB} \approx 3.65 \lesssim 3.82$ showing a gap between the stabilizer and quantum value, implying a separation on the sets of correlations $\mathcal C_{C,d} \subset \mathcal C_{\mathrm{STAB},d} \subset \mathcal C_{Q,d}$.

\subsubsection{On two stabilizer preparations \label{subsec:two_stabilizer}}
Then, one can ask how much non-stabilizerness is needed in the states in order to saturate the quantum value of the inequality. For the rest of this section, we will explore this question. In fact, define the $2-\mathrm{STAB}$ scenario in which two of the three states are restricted to be stabilizers and the third one can be any quantum state. Then the problem reduces to studying 
\begin{equation}
    S^{2-\mathrm{STAB}}_3(\rho) = \max_{\rho_1, \rho_2 \in \mathrm{ext}(\mathrm{STAB}_2)} \|\vec{r}_1 + \vec{r}_2 - \vec{r}\|_2 + \|\vec{r}_1 - \vec{r}_2\|_2 \;,
\end{equation}
as a function of the non-stabilizerness of $\rho$. Here we considered $\rho_{x \in \{1,2\}}  = (1+ \vec r_x \cdot \vec \sigma)/2$ and $\rho =(1+ \vec r \cdot \vec \sigma)/2$. Note also that the optimization is over the extrema of $\mathrm{STAB}_2$, where $\mathrm{ext}(\mathrm{STAB}_2)$ correspond to the vertices of the qubit stabilizer polytope:
\begin{equation}
    \mathrm{ext}(\mathrm{STAB}_2) = \{\frac 1 2 (1\pm X), \frac 1 2(1 \pm Y), \frac 1  2(1\pm Z)\}\;.
\end{equation}
Without loss of generality, we can always take $\mathrm{Tr}(\rho_1 \rho_2) =(1+ \vec r_1 \cdot \vec r_2)/2 = 1/2$, since orthogonal vectors still maximize the argument of the optimization. One can verify that  the maximum quantum value is reached if the argument is optimized:
\begin{equation}
    \max_{\rho \in \mathcal D(\mathbb C^2)} S^{2-\mathrm{STAB}}_3(\rho) =S^Q_3=1+2\sqrt 2 \;,
\end{equation}
which can be verified to be saturated for example, by $\vec r_1 = (1,0,0)$, $\vec r_2=(0,1,0)$ and $\vec r_3 = (-1/\sqrt 2, -1/\sqrt 2, 0)$. In fact, we can show a stronger result: This inequality self-tests $H$-type non-stabilizer states \cite{bravyi2005universal}, defined as the states in the Clifford orbit of $H = [1+(X+Z)/\sqrt 2]/2$:
\begin{prop}
     $S^{2-\mathrm{STAB}}_3(\rho) =S^Q_3$ if and only if $\rho = 
 C H C^\dagger$, for some $C \in \mathcal C\ell_2$.\label{prop: self_testing}
\end{prop}
\begin{proof}
    We can use Eq. (\ref{eq: optimalMS_S3}), where our vector to be optimized is now
    \begin{equation}
        \vec r =-\frac{\vec r_1+ \vec r_2}{\sqrt 2}\;,
    \end{equation}
    where $\vec r_1\cdot \vec r_2=0$, corresponding to $\rho_1$ and $\rho_2$ two orthogonal stabilizer states. From the 6 stabilizer states, there are $6 \times 4/2=12$ such pairs, which correspond exactly to the $12$ $H$-type non-stabilizer states.
\end{proof}
This guarantees the self-test of non-stabilizer states by using a prepare-and-measure witness.

One can ask about how much non-stabilizerness is needed to violate this inequality. The natural measure to consider is the non-stabilizerness by trace distance, defined as ~\cite{junior2025geometricanalysisstabilizerpolytope}
\begin{equation}
    \mathcal N \mathcal S(\rho)\equiv \min_{\sigma \in \mathrm{STAB}_2} \frac 1  2\|\rho-\sigma\|_1\;,
    \label{eq:trace_dist}
\end{equation}
plotted on Fig. \ref{fig: MTD} for visual aid. Geometrically, it represents the distance to the stabilizer polytope as measured by the trace norm, defined as $\|A\|_1 \equiv \mathrm{Tr}\sqrt{A^\dagger A}$. It is well-known that it is maximized by $T$-type non-stabilizer states, defined as the states in the Clifford orbit of $T = \frac 1 2(1+ \frac{1}{\sqrt 3}(X+Y+Z))$. 
\begin{figure}
    \centering
    \includegraphics[width=1.\linewidth]{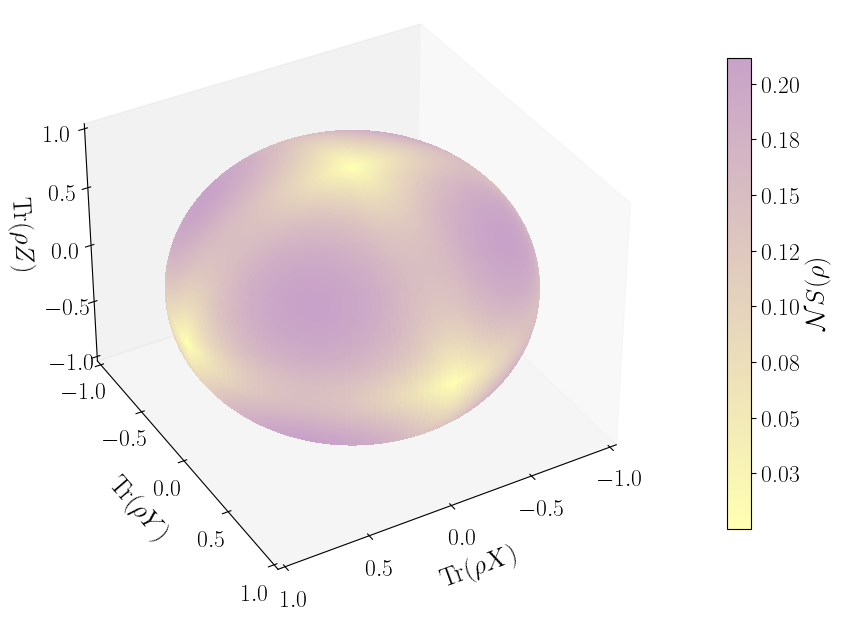}
    \caption{\textbf{Colourmap of $\mathcal{NS}$ of the Bloch Sphere surface.} We plot the value of the non-stabilizerness by trace distance. The center of the whiter zones corresponds to the stabilizer states with zero non-stabilizerness, while the center of the purple circles corresponds to rotations of the state $T$.}  
    \label{fig: MTD}
\end{figure}

In Ref.~\cite{junior2025geometricanalysisstabilizerpolytope}, it was derived the following analytical lower bound for the non-stabilizerness by trace distance:
\begin{equation}
    \mathcal N \mathcal S(\rho) \geq  -\frac{\sqrt 3}{3} \mathcal W^{\mathrm{STAB_2}}(\rho)\;,
\end{equation}
where we defined $\mathcal W^{\mathrm{STAB_2}}(\rho) \equiv\min[0, \{\mathrm{Tr}(\rho A^{\vec q})\}_{\vec q \in \mathbb \{0,1\}^3}] $, with $\{A^{\vec q}\}_{\vec q \in \mathbb \{0,1\}^3}$  the set of operators defining the eight facets of $\mathrm{STAB}_2$ as shown in Eq.~(\ref{eq: H-rep}).  Noting that the set of facets are related by Clifford unitaries, it shows that $\mathcal W^\mathrm{STAB}_2(\rho)$ is a faithful resource measure. In fact, for qubits, the quantity $\mathcal{W}^{\mathrm{STAB}_2}(\rho)$ assigns the same non-stabilizerness value to all states on a hyperplane parallel to a stabilizer polytope facet. This leads to an underestimation of non-stabilizerness for states outside the facet's projection, effectively measuring their distance as if they were in fact parallel to a facet. As such, $\mathcal{W}^{\mathrm{STAB}_2}$ acts as a lower bound to the non-stabilizerness measure $\mathcal{NS}$. Within the facet's projection, however, both measures coincide up to a rescaling by $-\sqrt{3}/3$, since the hyperplane distance matches the minimum distance from the states to the polytope\footnote{One can think of it as a coarsed-grained version of the trace distance that neglects the finiteness of the facets of the stabilizer polytope.}.  In particular, this is true  for the $T$-type non-stabilizer states \cite{bravyi2005universal},  giving the maximum of $\mathcal W^{\mathrm{STAB_2}}_\mathrm{max} \approx 0.366$, or $\mathcal{NS}_{\max}\approx 0.211$.

As pointed out by Proposition \ref{prop: self_testing}, the states that maximize $S^{2-\mathrm{STAB}}_3$ are the $H$-type non-stabilizer states, which maximize neither $\mathcal{NS}$ nor $\mathcal W^{STAB_2}$. It is interesting to study the value $S^{2-\mathrm{STAB}}_3(\rho)$ as a function of $\mathcal W^{\mathrm{STAB}_2}$ .  One can formulate the corresponding optimization problem 
\begin{align}
    s_3(w) \equiv &\max_{\rho \in \mathcal D(\mathbb C^2)} S^{2-\mathrm{STAB}}_3(\rho)\;, \label{eq:max_2stab}\\
    &\text{s.t. } \mathcal W^\mathrm{STAB_2}(\rho) =w, \nonumber
 \end{align}
which in terms of the Bloch vectors it can be rewritten as
\begin{align}
    s_3(w) \equiv &\max_{\vec r=(x,y,z) \in S^2}\|\vec{r}_1 + \vec{r}_2 - \vec{r}\|_2 + \|\vec{r}_1 - \vec{r}_2\|_2\;,\\
    & \text{s.t. } \vec r_1 =(1,0,0) \;,\; \vec r_2 = (0,1,0),\\
    & \frac 1 2\min[0,\{1+{q_1}x + q_2y +{q_3}z\}_{\vec q \in \mathbb \{\pm 1 \}^3}] =w,
\end{align}
where we already restricted to the optimal stabilizer states $\rho_1= (1+X)/2$ and $\rho_2 = (1+Y)/2$. The corresponding result for the inequality value is shown in Fig. \ref{fig: S3_magic}. The  peak at $w/\mathcal W^{\mathrm{STAB_2}}_\mathrm{max}= 0.565$  corresponds to the $H$-type non-stabilizer states, which indeed reaches the quantum maximum $S^Q_3 = 1+ 2\sqrt 2$. For larger violations, the optimal strategy cannot be obtained, and the value of the inequality again decreases to the value obtained by the $T$-type non-stabilizer states.  Finally, for completeness, we consider the violation of $S_3$ when all states are allowed to have the same non-stabilizerness. Since the optimization considers more states (not only the stabilizers), there exist several values of $w$ for which the optimal quantum strategy can be constructed, achieving the maximal quantum value $S_3^Q$. This is shown in the inset of Fig.~\ref{fig: S3_magic}. 

\begin{figure}
    \centering
    \includegraphics[width=1\linewidth]{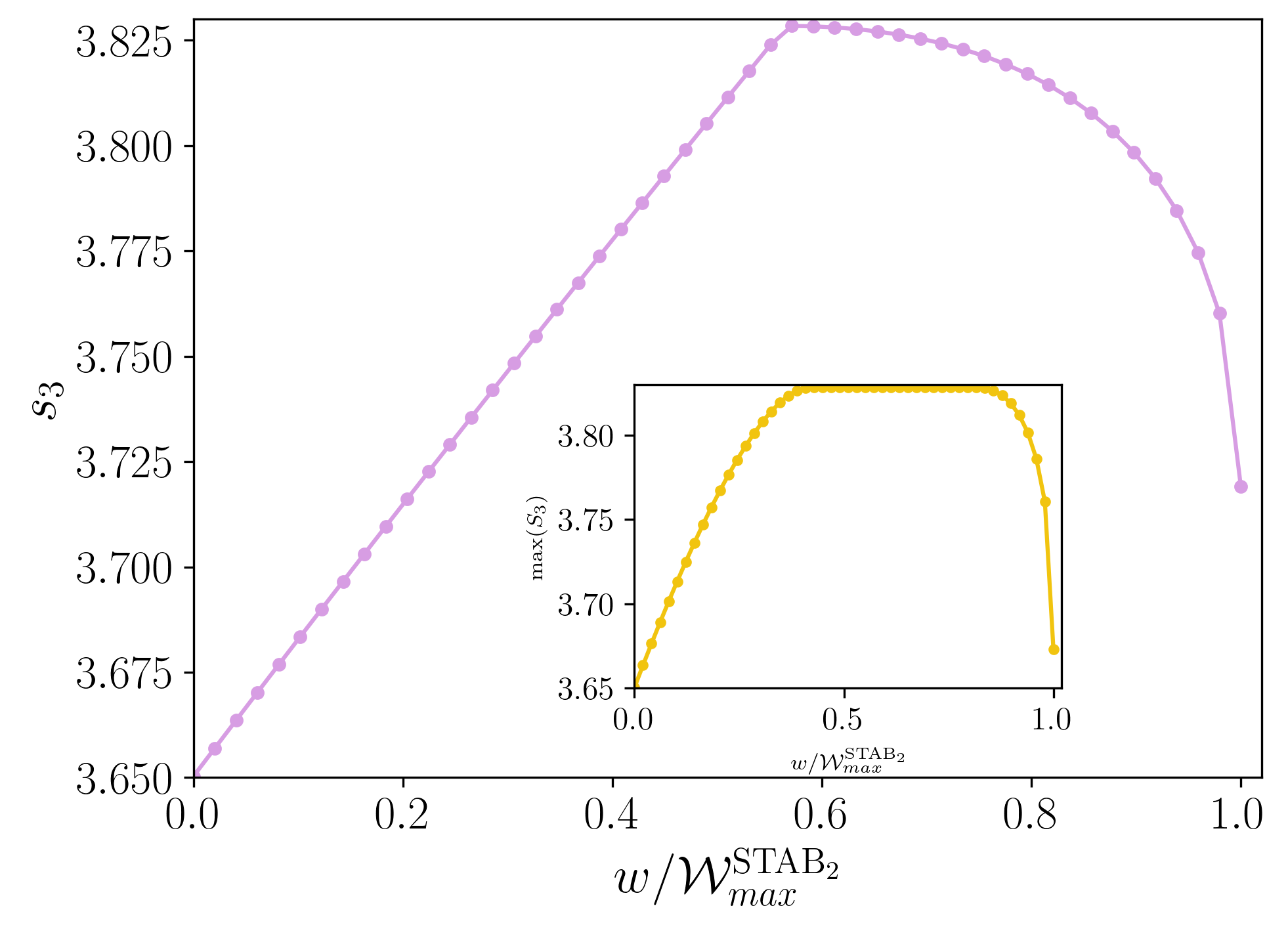}
    \caption{\textbf{Violation of $S_3$ as a function of  non-stabilizerness.}  In the main plot, we vary the non-stabilizerness of one state while keeping the other two as stabilizers. The maximum violation of $S_3$ occurs for the $H$ state with $w/\mathcal W^{\mathrm{STAB_2}}_\mathrm{max}(H) \approx 0.565$. As non-stabilizerness increases, the quantum bound is no longer reached due to a reduced set of compatible states. The final point corresponds to $w/\mathcal W^{\mathrm{STAB_2}}_\mathrm{max}(T) =1$ and $s_3(w) = 3.719$. The inset shows $S_3$ as a function of the non-stabilizerness when all three states belong to the same hyperplane $w$. A range of non-stabilizerness values still achieves the quantum maximum, but $S_3$ decreases for highly non-stabilizer states, reaching $s_3 = 3.6729$ for rotated $T$ states. 
    } 
    \label{fig: S3_magic}
\end{figure}

\subsection{The tilted inequality $S_3(t)$}
In this section we consider the recently proposed inequality~\cite{taoutioui2025assymetry} 
\begin{equation}
    S_3(t) = 2t(E_{11}+E_{21}-E_{31}) + 2(1-t)(E_{12}-E_{22}),\label{eq: tilted_s3}
\end{equation} 
defined by $t\in[0,1]$. In the same reference, the authors obtained the quantum value 
    \begin{equation}
        S_3^Q(t) = 2t + 4\sqrt{t^2 + (1-t)^2}
    \end{equation}
and the  classical value 
\begin{equation}
    S_3^C(t) =
    \begin{cases}
        4-2t\quad &\text{if}\quad t <1/2\\
        6t \quad  &\text{if}\quad 1/2\leq t
    \end{cases}.  
\end{equation}
For consistency we have multiplied by 2 the quantities defined in~\cite{taoutioui2025assymetry}. Note that for $t=1/2$ the above expressions reduced to the standard inequality $S_3$ and the corresponding classical and quantum values.  We ask the question whether by considering this inequality  there exists a gap between the $2-\mathrm{STAB}$ and the $\mathrm{STAB}$ scenarios for some values of $t$. By following a similar approach to that of the previous section we can state the following theorem. 
\begin{thm}
    \label{thm:tilted_witness_bounds}

Consider the tilted witness defined by Eq.~(\ref{eq: tilted_s3}) for the $(|\mathbb{X}|=3, |\mathbb{Y}|=2, |\mathbb{B}|=2)$ prepare-and-measure scenario, defined for a tilting parameter $t \in [0, 1]$ and qubit preparations $\{\rho_1,\rho_2, \rho_3\}$ .
The maximum achievable values for different classes of qubit preparations are given as follows:
\begin{enumerate}
    \item[(i)] \textbf{The Stabilizer Bound:} If all preparations are restricted to be stabilizer states, the maximum value is:
    \begin{equation}
        S_3^{\mathrm{STAB}}(t) = \max\left(6t, \ 4-2t, \ 2t\sqrt{5} + 2(1-t)\sqrt{2}\right)
    \end{equation}

    \item[(ii)] \textbf{The 2-Stabilizer Bound:} If at most one preparation is a non-stabilizer (magic) state, the maximum value is:
    \begin{equation}
        S_3^{\mathrm{2-STAB}}(t) = \max\left(6t, \ 4-2t, \ 2t+2\sqrt{2}\right)
    \end{equation}

    \item[(iii)] \textbf{The Quantum Bound:} This bound is already achievable with two non-stabilizer states.
\end{enumerate}
\end{thm}

\begin{proof}
   \textit{(i)} Consider the form of the states and measurements given in Eq.~(\ref{eq: states_measurements}), then by optimizing over the measurements as done in \cite{PhysRevResearch.2.043106}, one can rewrite this inequality as 
\begin{equation}
 S_3(t) = 2t\|\vec{r}_1 + \vec{r}_2 -\vec{r}_3\|_2  + 2(1-t)\|\vec{r}_1-\vec{r}_2\|_2.
\end{equation} 
Since $S_3(t)$ is linear on the states, the maximization is carried on the extremal set of the stabilizer polytope. Therefore, it is enough to consider  $\vec{r}_1$, $\vec{r}_2$, and $\vec{r}_3$ as the bloch vectors of the vertices of $\mathrm{STAB}_d$. Let us focus on the states $\vec{r}_1$ and $\vec{r}_2$. There are three cases:
\begin{itemize}
    \item $\vec{r}_1=\vec{r}_2$.\\ This implies
    \begin{equation}
         S_3(t) =2 t\|2\vec{r}_1  -\vec{r}_3\|_2,  
    \end{equation}
    which is maximized if we take $\vec{r}_3 = -\vec{r}_1$. Therefore $S^{\max_1}_3(t) = 6t$.

    \item $\vec{r}_1=-\vec{r}_2$.\\ This implies
    \begin{equation}
         S_3(t) = 2t\|-\vec{r}_3\|_2  +2(1-t)\|2\vec{r}_1\|_2,
    \end{equation}
    giving the maximum  $S^{\max_2}_3(t) = 4-2t  $
    \item $\vec{r}_1\perp\vec{r}_2$.\\ This corresponds to the optimal strategy found for the original inequality $S_3(1/2)$, $\vec{r}_1\perp\vec{r}_2$ and $\vec{r}_3 =-\vec{r}_1$,
    giving the maximum  $S^{\max_3}_3(t) = 2t\sqrt{5}+2(1-t)\sqrt{2}$.
    Therefore the stabilizer bound follows. 
\end{itemize}
\textit{(ii)} As for the $2-\mathrm{STAB}$ value, we now consider the first two preparations to be stabilizer states, while the third one is free to be any state. Again we have the three cases mentioned above. However, the first two cases result in the same maximum, attained by a pure state $\rho_3$. The third case changes and now we have:
\begin{itemize}
        \item $\vec{r}_1\perp\vec{r}_2$.\\ 
        The inequality $S_3(t)$ achieves its maximum for  $\vec{r}_3 = -(\vec{r}_1 + \vec{r}_2)/\|\vec{r}_1 + \vec{r}_2\|_2 = -(\vec{r}_1 + \vec{r}_2)/\sqrt{2}$, resulting in 
        \begin{align}
             S_3(t) &= 2t(\|\vec{r}_1 + \vec{r}_2\|_2 +1)  + 2(1-t)\|\vec{r}_1 - \vec{r}_2\|_2, \label{eq: s3t_two_Stabs}\nonumber\\
                    &= 2t(\sqrt{2} +1)  + 2\sqrt{2}(1-t).
        \end{align}
        Therefore $S^{\max_{3'}}_3(t) =  2t+2\sqrt{2}$.

\end{itemize}
 Which implies  the bound $S_3^{\mathrm{2-STAB}}(t)$  shown in the theorem. Finally, for \textit{(iii)}, the full quantum bound is achieved even when one preparation is restricted to be a stabilizer state. This is because the optimal value, derived from Eq. (\ref{eq: s3t_two_Stabs}), depends only on the relative angle between two preparation vectors. Fixing one vector to be a stabilizer does not constrain the ability to achieve the optimal angle with a second, arbitrary quantum state. Therefore, the bound for one stabilizer and two magic states is identical to the full quantum bound.
\end{proof}
In the range $t\in[0,1]$, by comparing the different cases, we can express these bounds as 
\begin{equation}
    S_3^{\mathrm{STAB}}(t) =
    \begin{cases}
        4-2t\quad &\text{if}\quad t <t_0\\
        2t\sqrt{5}+2(1-t)\sqrt{2} \quad &\text{if}\quad t_0\leq t \leq t_1\\
        6t\quad &\text{if}\quad t_1<t
    \end{cases}.  
\end{equation}
with $t_0 = \frac{\sqrt{2}-2}{\sqrt{2}-\sqrt{5}-1}$, $t_1= \frac{\sqrt{2}}{3+\sqrt{2}-\sqrt{5}}$
and
\begin{equation}
    S_3^{\mathrm{2-STAB}}(t) =
    \begin{cases}
        4-2t\quad &\text{if}\quad t < 1- 1/\sqrt{2}\\
       2t+2\sqrt{2} \quad &\text{if}\quad 1- 1/\sqrt{2} \leq t \leq 1/\sqrt{2} \\
        6t\quad &\text{if}\quad 1/\sqrt{2}<t 
    \end{cases}.
\end{equation}
In Fig.~\ref{fig: S3_tilted} we plot these curves, and we show that there exist some values of $t$ where there is a gap between the classical, stabilizer, 2-stabilizer and quantum value. The shaded regions show the values for which the witnes $S_3(t)$ certifies the presence of  one non-stabilizer state (green area) and two non-stabilizer states (yellow area). 

\begin{figure}
    \centering
    \includegraphics[width=1\linewidth]{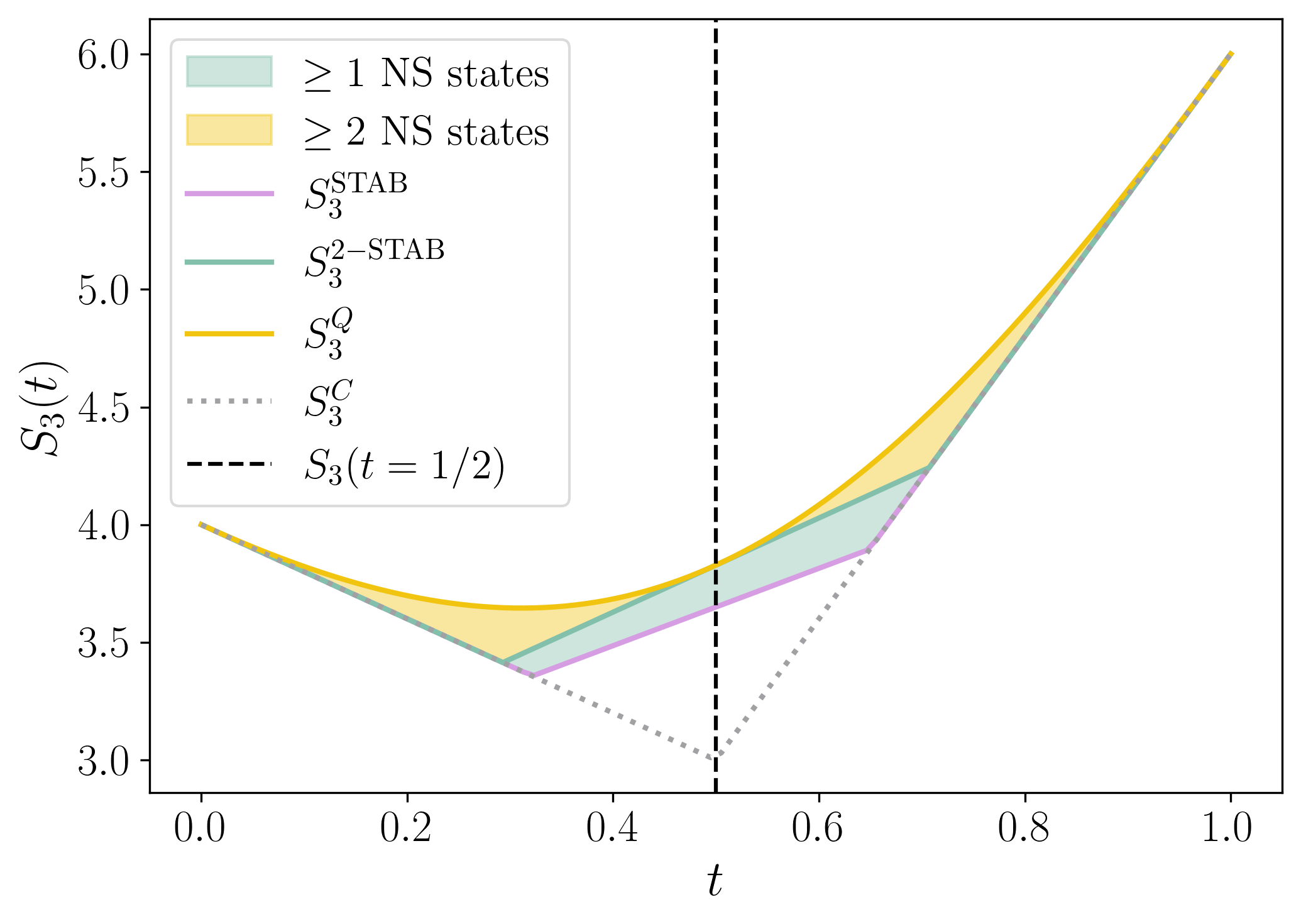}
    \caption{\textbf{Distinct bounds for the inequality $S_3(t)$}. The plot displays the Classical ($S_3^C$, dotted gray), Stabilizer ($S_3^{\mathrm{STAB}}$, purple), 2-Stabilizer ($S_3^{2-\mathrm{STAB}}$, green), and Quantum ($S_3^Q$, yellow) bounds. The 2-Stabilizer bound represents the maximum value achievable with at most one non-stabilizer (NS) state. The vertical dashed line at $t=1/2$ marks the standard, non-tilted witness. The shaded regions highlight the advantage of using NS states: a value in the light green area certifies at least one NS state, while a value in the yellow area certifies at least two. The clear separation of these bounds demonstrates that the tilted witness can provide a more fine-grained certification of non-stabilizerness.  }
    \label{fig: S3_tilted}
\end{figure}

\section{Quantum random access codes}\label{sEC:QUANTUM RANDOM ACCESS CODES}

Although it is now clear that a separation already appears even in the minimal PAM scenario, we will present a class of scenarios, with an increasing number of preparations with witnesses corresponding to a physically-motivated quantity.

\subsection{Qubit QRACS \label{sec:qubitqracs}}
The witnesses considered in this section arise in the study of quantum random access codes \cite{ambainis2009quantumrandomaccesscodes,PhysRevA.85.052308,tavakoli2015quantum}. In this scenario, there are two parties, Alice and Bob, where Alice encodes a bit-string $x \in \{0,1\}^N$ into a single qubit state and sends it to Bob, who extracts one of the $N$ bits of $x$ by performing a measurement. The encoding strategy are $2^N$ assignments of bit-strings into states, $x \in \{0,1\}^N \mapsto \rho_x \in \mathcal D(\mathbb C^2)$ and the decoding strategy is given by the choice of $N$ measurements $(M^1, \cdots, M^N)$. The average success probability of Bob correctly measuring the bit-string in such a setup is
\begin{equation}
    \bar p_N = \frac{1}{2^N N } \sum_{x \in \{0,1\}^N} \sum_{y=1}^N p(x_y|x,y)\;,
\end{equation}
where $p(b|x,y) =\mathrm{Tr}(\rho_x M^y_{b})$ is the probability of obtaining of the outcome $b$ of the measurement $M^y$ in the encoded state $\rho_x$. Those behaviors are clearly described by the set $\mathcal C_{Q,d}$ of a $(d=2, |\mathbb X|=2^N, |\mathbb Y|=N, |\mathbb B|=2)$ PAM scenario, which we will also refer as the $N \to 1 $ qubit QRAC \cite{ambainis2009quantumrandomaccesscodes}. The other two corresponding sets $\mathcal C_{C,d}$ and $\mathcal C_{\mathrm{STAB},d}$ are defined by restricting the encoding to classical and stabilizer states, respectively. Remarkably, as we will discuss below, the average success probability is a separating witness for various values of $N \geq 3$:
\begin{equation}
  \bar p^C_N < \bar p^\mathrm{STAB}_N < \bar p^Q_N.
\end{equation}
 First, note that $M^{b=1}_y = 1- M^{b=0}_y$, which allows us to write:
\begin{align}
    \bar p_N &= \frac 1 2 \left(1 + \frac{1}{2^{N-1} N} \sum_{x \in \{0,1\}^N} \sum_{y=1}^N (-1)^{x_y}p(0|x,y)\right)\nonumber\\
    & =\frac 1 2 \left(1 + \frac{T_N}{2^{N-1} N}\right)
\end{align}
where we have defined 
\begin{equation}
    T_N \equiv \sum_{x \in \{0,1\}^N} \sum_{y=1}^N (-1)^{x_y}p(0|x,y).
\end{equation}

The optimization over $T_N$ can be written in terms of Bloch  vectors as stated in the following Lemma:
\begin{lem}\label{lemma: ASP_QRACS}
    Let $ S \in \{C, \mathrm{STAB},Q\}$ be one of the three scenarios in the $N \to 1$ qubit QRAC. Then, the average success probability is given by
    \begin{equation}
        \bar p^S_N = \frac{1}{2} \left(1 + \frac{T_N^S}{2^{N-1} N}\right)\;,
    \end{equation}
    where
    \begin{equation}
        T^S_N = \max_{\{\vec r_x\} \in \mathrm{ext}(S)} \frac 1 2 \sum_{y=1}^N \|\sum_{x \in \{0,1\}^N}(-1)^{x_y} \vec r_x \|_2\;,
        \label{eq:TSN}
    \end{equation}
    and $\{\vec r_x\}$ are the Bloch vectors of the states with the encoded message. The extremal set $\mathrm{ext}(S)$ is defined by 
    \begin{equation}
        \mathrm{ext}(S) = 
        \left \{
        \begin{aligned}
            &\{e_3, -e_3\}\;, \mathrm{for\;} S=C\;,\\
            &\{\pm e_1,\pm e_2,\pm e_3\}\;, \mathrm{for\;} S= \mathrm{STAB}\;,\\
            &S^2 \;, \mathrm{for\;} S= Q\;,
        \end{aligned}
        \right.
    \end{equation}
    with $e_1 = (1,0,0)$, $e_2=(0,1,0)$ and $e_3= (0,0,1)$ .
\end{lem}
\begin{proof}
    We can write both the states and measurements in the Bloch sphere as:
    \begin{align}
        \rho_x &= \frac 1 2(1 + \vec r_x \cdot \vec \sigma)\;,\\
        M^0_y &= \frac 1 2(1 + \vec v_y \cdot \vec \sigma)\;,
    \end{align}
    yielding the following form for the optimized witness:
    \begin{align}
        T^S_N &= \max_{\substack{\{\rho_x\} \subseteq S\\ \{M_y\} \subseteq \mathcal M_{N}(\mathbb C^2)}}\sum_{x \in \{0,1\}^N} \sum_{y=1}^N (-1)^{x_y}\mathrm{Tr}(\rho_x M^0_y)\\
        &= \max_{\substack{\{\vec r_x\} \subseteq \mathrm{ext}(S)\\ \{\vec v_y\} \subseteq S^2}} \sum_{x \in \{0,1 \}^N} \sum_{y=1}^N \frac{(-1)^{x_y}}{2}(1+ \vec v_y \cdot \vec r_x)\;.
    \end{align}
    We now proceed to optimize over the measurements. The first term vanishes, due to the fact that every bit string $x \in \{0,1\}^N$ has a corresponding parity-flipped string $\mathrm{XOR}(x)$, canceling in pairs. Defining $\vec R_y = \sum_{x}(-1)^{x_y} \vec r_x$, we arrive at
    \begin{equation}
        T^S_N = \max_{\substack{\{\vec r_x\} \subseteq \mathrm{ext}(S)\\ \{\vec v_y\} \subseteq S^2}} \frac 1 2 \sum_{y=1}^N  \vec v_y \cdot \vec R_y\;,
        \label{eq:TSN_sum_of_inner}
    \end{equation}
    which is a sum of inner products, maximized when $\vec v_y = \vec R_y/\|\vec R_y\|_2$. By substituting in Eq. (\ref{eq:TSN_sum_of_inner}), we arrive at the result. 
\end{proof}
For $N=1$,  $|\mathbb X|= 2^1 =d$ and then due to Lemma \ref{lem:non-trivial} the scenario is trivial. Then, our task is to consider the optimization problem only for $N \geq 2$. To do so, we present a helpful proposition.
\begin{prop}
    Let $\{\vec r^*_x\}_x \in \mathrm{argmax} (T^S
    _N)$. Then, $\vec r^*_x=- \vec r^*_{\mathrm{XOR}(x)}$.
    \label{prop:XOR}
\end{prop}
\begin{proof}
    Consider a splitting of the set of bit strings of the form $\{0,1\}^N = S_+ \sqcup S_-$, where $S_- = \mathrm{XOR}(S_+)$. This is,  $S_-$ are the bitstrings of $S_+$ under the XOR operation. It follows that
    \begin{align}
        \vec R_y &= \sum_{x \in \{0,1\}^N} (-1)^{x_y} \vec r_x \\
        &= \sum_{x \in S_+}(-1)^{x_y} \vec r_x- \sum_{x \in S_+} (-1)^{x_y} \vec r_{\mathrm{XOR}(x)}\\
        &= \sum_{x \in S_+}(-1)^{x_y}[\vec r_x -\vec r_{\mathrm{XOR}(x)}]\;.
    \end{align}
    This configuration has a witness value of $T_N (\{\vec r_x\}_{x \in S_+ \sqcup S_-}) = 1/2\sum_y \| \vec R_y\|_2$. We express the 2-norm by the optimization
    \begin{equation}
        \|\vec R_y\|_2=\max_{\vec u_y \in S^2}\; \vec u_y \cdot \vec R_y\;,
    \end{equation}
    which allows us to write
    \begin{align}
        \nonumber & T^S_N(\{\vec r_x\}_x) \\ &= \max_{\{\vec u_y\} \in S^2} \frac 1 2 \sum_{x \in S_+} \left(\sum_{y=1}^N(-1)^{x_y} \vec u_y\right) \cdot (\vec r_x-\vec r_{\mathrm{XOR}(x)})\\
        &= \max_{\{\vec u_y\} \in S^2} \frac 1 2 \sum_{x \in S_+} \vec V_x \cdot (\vec r_x-\vec r_{\mathrm{XOR}(x)}) \;,
    \end{align}
    where we defined $\vec V_x = \sum_y (-1)^{x_y}\vec u_y$. The sum of inner products, upon optimizing over $\{\vec r_x\}_x$, is maximized by:
    \begin{equation}
        \vec r_x^* = -\vec r^*_x= \frac{\vec V_x}{\|\vec V_x\|}\;.
    \end{equation}
\end{proof}
For small $N$, we can show the following.
\begin{thm}
    The  set of inequalities
    \begin{itemize}
        \item $T^C_2 = 2 < T^\mathrm{STAB}_2 = T^Q_2=2\sqrt 2$,
        \item $T^C_3 = 6 < T^\mathrm{STAB}_3 = 2\sqrt 6 + \sqrt 2$,
    \end{itemize}
    are true.
    \label{thm:TN_stab_values}
\end{thm}
\begin{proof}
    The classical values of the witnesses were already derived in \cite{ambainis2009quantumrandomaccesscodes,PhysRevA.85.052308}. We will now proceed to compute the stabilizer values assuming the result shown in  Proposition \ref{prop:XOR}. For $N=2$, Eq.~(\ref{eq:TSN}) is simplified to
    \begin{align}
        T^\mathrm{STAB}_2 =\max_{\{\vec r_{00}, \vec r_{01}\} \in \mathrm{ext}(\mathrm{STAB}_2)}(&\| \vec r_{00} + \vec r_{01}\|_2 + \nonumber\\
        & \|r_{00}- \vec r_{01} \|_2)\;.
    \end{align}
    By defining $\vec r_+ = (\vec r_{00} + \vec r_{01})/\sqrt 2$ and $\vec r_- = (\vec r_{00}- \vec r_{01})/\sqrt 2$, as vectors of 2-norm 1,  it is clear that the optimal value is $\sqrt 2 (\|\vec r_+\|_2 + \| \vec r_{-}\|_2)\mapsto 2\sqrt 2$. Furthermore, note that this solution can also be obtained if the vectors are also to be optimized over the sphere, $\{\vec r_{00}, \vec r_{01}\} \in S^2$.
    
    For $N=3$, the objective function to be optimized is:
    \begin{equation}
        \sum_{\vec q \in \{-++,+-+,++-\}} \| \vec r_{000} + {q_1} \vec r_{100} + {q_2}\vec r_{010} + {q_3}\vec r_{001}\|_2\;,
    \end{equation}
    where we optimize over $\{\vec r_{000}, \vec r_{100}, \vec r_{010}, \vec r_{001}\} \in \{\pm e_1, \pm e_2, \pm e_3\}$. On this case, an exhaustive search over the configurations yields the solution (one of them displayed in the Appendix \ref{app: stab_optimal_strategies})
    \begin{equation}
        T^\mathrm{STAB}_3 = 2\sqrt 6 + \sqrt 2\;.
    \end{equation}
\end{proof}
The quantum value for $N=3$ was proven to be tight in Ref.~\cite{PhysRevA.98.062307} with the exact value of  $T^Q_3 =4\sqrt{3} \approx 6.928$, larger than $T^\mathrm{STAB}_3\approx 6.313$ computed exactly above. Hence, there is a gap between stabilizer and general qubit states for QRACs with $N=3$. Note that this gap is exact; to obtain it, we considered all the possible combinations of stabilizer vertices. For larger $N$, the exhaustive method among the $O(\exp(\exp N))$ options becomes unfeasible and an alternative formulation of the problem is needed. As done in Eq.~(\ref{eq:max_2stab}), consider:
\begin{align}
    t_N(w) \equiv &\max_{\{\rho_x\}_x \in \mathrm{ext}(\mathrm{STAB}_2)} T_N(\{\rho_x\}_x)\;.\\
    &\mathcal W^{\mathrm{STAB_2}}(\rho)=w \nonumber
\end{align}
Note that $T^{\mathrm{STAB}}_N = t_N(0)$. In terms of the Bloch vectors, we obtain:
\begin{align}
    t_N(w) \equiv &\max_{\{\vec r_x= (r_{x,1}, r_{x,2}, r_{x,3})\}_x \in S^2}  \frac 1 2 \sum_y \| \sum_x (-1)^{x_y} \vec r_x \|_2\;,\\
    & \frac 1 2\min[0,\{1+\sum_{j=1}^3 q_i r_{x,i}\}_{\vec q \in \mathbb \{\pm  \}^3}] =w\;,\;\forall x \in \{0,1\}^N
\end{align}
which allows us to obtain also the numerical values for $N=4$ and $N=5$. Table \ref{tab: TN_values} shows the results compared to the classical and quantum bounds.
\begin{table}[h]
    \centering
    \begin{tabular}{cccc} % 'l' for left-aligned, 'c' for center-aligned
        \toprule
        $N$ & $T_N^C$ &$T_N^\mathrm{STAB}$ & $T^Q_N$ \\
        \midrule
        $2$ &  $2$ & $2\sqrt 2$ & $2 \sqrt 2$ \\
        $3$ & $6$ & $2 \sqrt 6 + \sqrt 2 \approx 6.313 $ & $\approx 6.928$\\
        $4$ & $12$ & $10 \sqrt 2 \approx 14.142$  &$\approx 15.458$\\
        $5$ & $30$ & $16 \sqrt 2 + 8 \approx 30.627$  &$\approx 34.172$\\
        \bottomrule
    \end{tabular}
    \caption{\textbf{Maximum values of $T_N$ restricted to the three different scenarios.} The values of the classical scenario are computed analytically, and in the quantum scenario are estimated numerically. These can be found in \cite{PhysRevA.98.062307,ambainis2009quantumrandomaccesscodes,PhysRevA.85.052308}. For $N=2,3$, the exact stabilizer value is exact as stated in Theorem \ref{thm:TN_stab_values} while for $N=4,5$ it was found numerically by computing $t_N(0)$.}
    \label{tab: TN_values}
\end{table}

 Interestingly, for $N \geq 3$, a hierarchy between the scenarios is found. For $ w >0$, and assuming that all prepared states have the same degree of non-stabilizerness, the corresponding stabilizer optimization problem is shown in Fig. \ref{fig: T_magic}. For $N=2$, we see that it always reaches the quantum value except when it reaches the $ T$-type states. When restricted to states with high non-stabilizerness, the optimal strategy maximally violating the corresponding inequality cannot be found.  For $N \geq 3$, we see that after some violation value $w_c$ the inequality reaches its maximum while for $N =4$, the witness has a similar behaviour as $s_3(w)$ (see \eqref{eq:max_2stab}), where the maximal violation is found at intermediate regions of the violation and it decreases until $w= \mathcal{W}^\mathrm{STAB_2}_\mathrm{max}$.

The above results can be explained geometrically, since the corresponding quantum-optimal states are known, and their geometry is extensively discussed in \cite{ambainis2009quantumrandomaccesscodes}. As the proof of Theorem \ref{thm:TN_stab_values} illustrates, the optimizing solution for $N=2$ is represented by a square in the bloch sphere, which can be unitarily rotated by the symmetry indicated in Proposition \ref{prop:unitary_symmetry} into $\vec r_{00} = (1,0,0)$, $\vec r_{01} = (0,1,0)$, $\vec r_{11} = (-1,0,0)$ and $\vec r_{10} = (0,-1,0)$, which are points in the stabilizer polytope. Therefore, it does not present a gap when restricting to stabilizer states. However, note that the optimal strategy, corresponding exactly to the Hamming cube of $\{0,1\}^2$ inscribed into the stabilizer octahedron, cannot be reproduced by $T$  non-stabilizer states since no pair of these states are antipodal, therefore diminishing the violation of $T_2$ when restricting to high-non-stabilizer states. For $N=3$ we have a different behavior: the maximal quantum value is obtained for the $\{0,1\}^3$ Hamming cube, which cannot be included in the stabilizer octahedron. Therefore, there is a gap between the stabilizer and quantum values. However, the eight $T$-type maximal non-stabilizer states lie indeed at the coordinates of the $\{0,1\}^3$ Hamming cube. This allows for the optimal strategy and therefore for every value of non-stabilizerness the witness can obtain the maximal quantum value $T^Q_3$.  As mentioned before,  in Fig. \ref{fig: T_magic} we observe that this behavior is robust, since the maximum quantum value is obtained after some $w_c$. For $N \geq 4$, the known quantum solutions break this cubic symmetry (since the higher Hamming cubes do not even make sense in three-dimensional space) and the states with maximal facet violation do not necessarily correspond to the states with maximal values of $T_N$.

\begin{figure}
    \centering
    \includegraphics[width=1\linewidth]{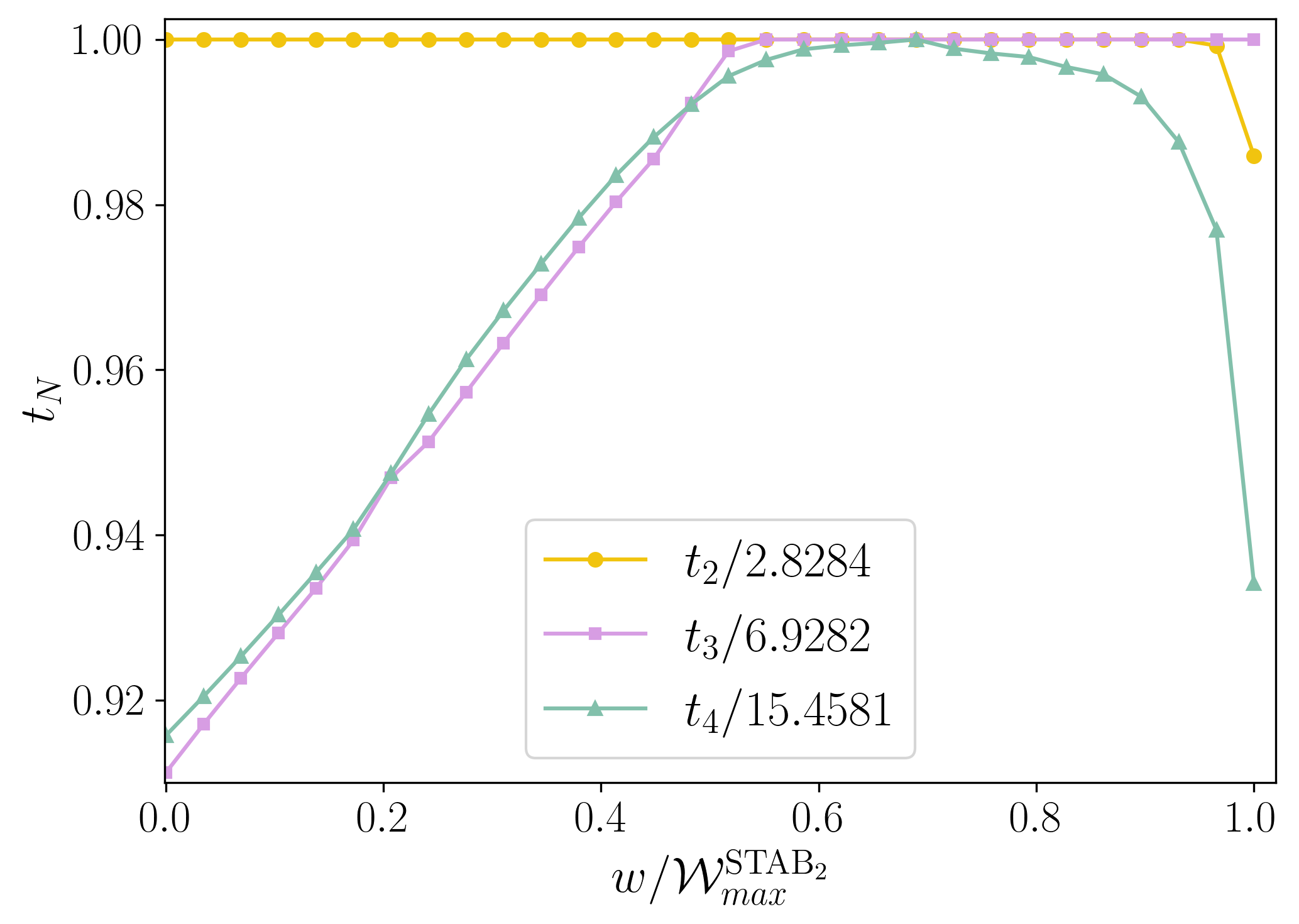}
   \caption{\textbf{Value of the witness of $T_N$ as a function of the non-stabilizerness of all $2^N$ states.}  We show the normalized (by its maximum) value of the witnesses $T_N$ . We assumed that all $2^N$ states have the same non-stabilizerness $w$.  For $N=2$ (the curve with golden circles), we show that for any positive value $w$, there exists a combination of states, all of them with non-stabilizerness $w$ that saturates the quantum bound. However, when restricting to states very close to the maximal non-stabilizerness, we cannot obtain the optimal strategy. The maximum for this point is $T_2 = 2.7886$. We also plot the witness $T_3$ (purple squares). In this case we see that after reaching the maximum around $w/\mathcal W^{\mathrm{STAB_2}} \approx 0.539$, the quantum maximum is still saturated; the reason is explained in the main text. Finally, we show the witness $T_4$ (green triangles). Due to the dimension of the optimization space, escaping local maxima is difficult and at some points we cannot completely assure that it is the global maximum, although we conjecture it is very close. }
    \label{fig: T_magic}
\end{figure}

As done in Sec. \ref{subsec:two_stabilizer}, one can also look at the value of $T_N$ under intermediate scenarios, where part of the prepared states is stabilizer and part is allowed to be arbitrary,  On Table \ref{tab: T3_Stab_values}, we show the results for $N=3$ and $N=4$.

\begin{table}[h]
    \centering
    \begin{tabular}{ccc} % 'l' for left-aligned, 'c' for center-aligned
        \toprule
        $s$ stabilizers & $T_3$ &$T_4$ \\
        \midrule
        $0$ &  $T_3^Q = 6.9282$ & $T_4^Q = 15.4581$ \\
        $1$ & $6.9282 $ &$15.4581$\\
        $2$ & $6.9068$  &$15.4581$\\
        $3$ & $6.8539$  &$15.4581$\\
        $4$ & $6.6071$  &$15.4581$\\
        $5$ & $6.6038$  &$15.3664$\\
        $6$ & $6.5983$  &$15.1676$\\
        $7$ & $6.4931$  &$14.9220$\\
        $8$ & $T^{STAB}_3 =6.3132$  &$14.6377$\\
        \bottomrule
    \end{tabular}
    \caption{\textbf{$T_N$ maximum reachable value dependent on the number of stabilizer preparations.} The table shows the maximum value attainable by $T_3 \text{ and } T_4$  when considering $s$ stabilizer preparations. They are found by brute forcing through all the $6^s$ options of stabilizer states and numerically optimizing the Bloch sphere of the remaining $2^N-s$ states. In the case of $T_3$, the value $T^{STAB}_3 = 2\sqrt{6}+ \sqrt{2} \approx 6.3132 $ is the exact stabilizer value obtained by considering all the $6^8$ combinations of stabilizer states. For  $N>3$ the stabilizer value becomes difficult to find using this method. }
    \label{tab: T3_Stab_values}
\end{table}

\subsection{PAM scenario with $d$ dimensional systems} \label{SUBSEC:PAM scenario with $d$ dimensional systems}

For higher-dimensional systems, one needs to consider other QRACs. A class of PAM scenarios labelled by $(d, |\mathbb X| =d^2, |\mathbb Y| = 2, |\mathbb B|=d)$ allow for the so-called $2^{(d)} \to 1 $ QRAC \cite{Farkas2019} with the corresponding average success probability
\begin{equation}
    \bar p_d = \frac{1}{2d^2} \sum_{x, x^\prime \in [d]}\mathrm{Tr}[\rho_{x}(M^1_x+M^2_{x^\prime})],
\end{equation}
where we  denote $[d] = \{1, \cdots, d\}$. In this scenario, on the preparation side, Alice gets two uniformly random inputs $x,x^\prime \in \{1, \cdots, d\}$ which allow her to label $d^2$  $d$-dimensional preparations $\rho_{x, x^\prime} \in \mathcal D(\mathbb C^d)$. Alice sends to Bob one of these states and tries to guess Alice's input. That is, given that Bob receives an input $y\in\{1,2\}$, if $y=1$ he must guess the value $x$. Otherwise, if $y=2$ he would have to guess $x^\prime$ instead. To guess  Alice's input $x$, Bob performs a measurement described by the POVM $\{M^1_x\}_{x=1}^d$, and when trying to guess $x^\prime $ he uses the POVM  $\{M^2_{x^\prime}\}_{x^\prime=1}^d$. As POVM's, they are  positive semi-definite operators satisfying $\sum_{x=1}^d M^1_x=\sum_{x^\prime=1}^d M^2_{x^\prime} = \mathbb{I}$. 

We will now use $\bar p_d$ as a witness. In \cite{tavakoli2015quantum}, it was found the classical bound to be $p^C_d = (1+1/d)/2$, and
in \cite{Farkas2019}, the corresponding quantum maximum is shown to be:
\begin{equation}
p_d^Q  =\frac{1}{2}(1+\frac{1}{\sqrt d}),
\end{equation}
only attained if Bob's measurements are all rank-$1$ projective measurements and mutually unbiased, while  the prepared states $\rho_{x,x^\prime}$ correspond to the highest eigenvalue of $M^1_x+ M^2_{x^\prime}$. Following the strategy of the previous sections, we first frame the computation of the witness value as the following optimization problem

\begin{lem}
    Let $\{A^{\vec q} \}_{\vec q }$ be the facet representation of $\mathrm{STAB}_d$. Then:
    \begin{align}
        \bar p_d^\mathrm{STAB} = &\max_{\substack{\{\rho_{x,x^\prime} \}\in M_{d \times d}(\mathbb C)\\ \{M^1, M^2\} \in  M_{d \times d}(\mathbb C)}} \bar p_d(\{\rho_{x,x^\prime}\},\{M^1, M^2\})\\
        &\rho_{x,x^\prime} \succeq 0 \quad ;\quad \mathrm{Tr}(\rho_{x,x^\prime})=1 \quad \forall x,x^\prime \in [d]\\
        &\mathrm{Tr}(\rho_{x,x^\prime} A^{\vec q}) \geq 0 \quad \forall \vec q\\
        &\sum_{b=1}^d M^y_b = 1 \quad\quad \forall y \in \{1,2\}\\
        & M^y_b \succeq 0 \quad \forall b \in [d] ,\quad \forall y \in \{1,2\}
    \end{align}
    where $M_{d \times d}(\mathbb C)$ denotes the set of $d \times d$ complex matrices.
\end{lem}
\begin{proof}
    These results follow from the semidefinite characterization of the set of states and measurements:
    \begin{align}
        \mathcal D(\mathbb C^d)  &= \{\rho \in M_{d \times d}(\mathbb C)\;|\; \mathrm{Tr}\rho= 1\;,\; \rho \succeq 0\},\\
        \mathcal M_d(\mathbb C^d)  &= \{\{M_b\}_{b=1}^d \in M_{d \times d}(\mathbb C)\;|\; \sum_b M_b=1\;,\; M_b \succeq 0\},
    \end{align}
    and the facet representation of the stabilizer polytope, already written on Eq. (\ref{eq: H-rep}).
\end{proof}

If either the measurements or the states are held fixed, the problem becomes a semi-definite program (SDP), amenable to standard optimization techniques \cite{skrzypczyk2023semidefinite}. However, if both of them are to be optimized, the corresponding objective function is quadratic on matrices with conic constraints. A common solution that allows writing it as an SDP is by performing a seesaw algorithm at the risk of not finding the global maxima~\cite{Mironowicz_2024}.

For $d=2$, note that this corresponds to the $(d=2, |\mathbb X| =4, |\mathbb Y| =2, |\mathbb B|=2)$ quantum scenario, which is exactly the one we just discussed on Sec. \ref{sec:qubitqracs}. The first non-trivial case appears for $d=3$, where we bound qutrit correlations. By doing this seesaw procedure, we find that the stabilizer value is 
\begin{equation}
p_3^\mathrm{STAB}\approx0.712< p_3^{Q} = (1+1/\sqrt{3})/2 \approx0.789,
\end{equation}
with $p^\mathrm{STAB}_3 > p^C_3 =(1+1/3)/2 =2/3 \approx 0.667 $. To obtain this value, we initialized the seesaw algorithm with different random states and measurements, and after $10^5$ realizations, we could not find a higher value.

To complete our analysis, we introduce in the next subsection a formal criterion that shows necessary conditions for the existence of a stabilizer bound for a given witness. The key idea is that, when restricting to stabilizer states, the set of pure stabilizer states is discrete and of measure zero within the set of all quantum states. As a result, the maximum value of any witness is generally unattainable by stabilizer states, unless the optimal quantum strategy happens to align (up to rigid rotations) with one of the polytope's vertices.

\section{Separating Witnesses and State overlaps}\label{sec:Separating Witnesses}

In the previous sections, we introduced separating witnesses for various PAM scenarios, illustrating the hierarchy established in Eq.~(\ref{eq:strict_hierarchy}). Here, we present a necessary condition to determine whether a given witness with coefficients $W_{x,y}^b$ may certify non-stabilizerness. Our analysis draws upon recent developments that characterize resource sets through state overlaps \cite{PhysRevA.109.032220, PhysRevA.101.062110}.

Due to convexity, the optimization given in Eq.~(\ref{eq:state-measurement_opt}) can be performed over the extremal sets (of measurements and states, although in the following we will look only at the set of states) defined in the $S$ scenario. Denoting  $\mathrm{ext}S$ the extremal set of states in scenario $S$, Eq.~(\ref{eq:state-measurement_opt}) is equivalent to 
\begin{equation}
    W^S =\max_{\substack{|\psi\rangle \in \mathrm{ext}S\\ \{M^y\} \in \mathcal M_\mathbb B(\mathbb C^d)}} \sum_{b,x,y}W^{b}_{x,y}\bra{\psi}M^y_b\ket{\psi}\;,
\end{equation}
In particular, if $S$ is the quantum scenario, the extremal set of states are the pure states $\mathrm{ext}S = \mathbb C^d$.  For $S \in\{C, \mathrm{STAB}\}$, the corresponding set of states is not only convex, but is also given by a polytope. This means that $\mathrm{ext} S$ is a finite set of extremal elements. In the classical scenario the number of these elements is given by $|\mathrm{ext} C|=d$, where $|\mathrm{ext} C|$ denotes the cardinality of a $d$-dimensional orthonormal basis, while in the stabilizer scenario the number of extremal elements is given by $|\mathrm{ext}\mathrm{STAB}| = d(d+1)$,  $|\mathrm{ext}\mathrm{STAB}|$ denotes the number of vertices considered in Eq.~(\ref{eq: StabPolytope_Vrep}). We claim that the natural object to look at is the corresponding Gram Matrix of the extremal states
\begin{equation}
    G^S_{ij} \equiv \langle \psi_i|\psi_j\rangle  \;;\; i,j \in \{1, \cdots,|\mathrm{ext}S|\}\;,
\end{equation}
which is invariant under unitaries $\ket{\psi_i} \mapsto U\ket{\psi_i}$. Indeed, the following Proposition shows that the Gram matrix completely determines the structure of unitary invariants. 
\begin{prop}\label{prop:unitaries_overlaps}
    Consider two sets of states $\{\ket{\psi_x}\}_x, \{\ket{\phi_x}\}_x \subseteq \mathbb C^d$ with same cardinality. Then, there exists a unitary $ U \in \mathcal U(d)$ such that $\ket{\psi_x} = U \ket{\phi_x}$ if and only if :
    \begin{equation}
    \langle \psi_x|\psi_{x^\prime}\rangle =\langle \phi_x |\phi_{x^\prime}\rangle \;, \forall x,x^\prime\;.
    \end{equation}
\end{prop}
This proposition is a corollary of Wigner's theorem \cite{SIMON20086847}. Recalling the invariance under unitaries of the witness, $W(\{U |\psi_x\rangle \langle |\psi_x|U^\dagger\}_{x \in \mathbb X}) =W(\{|\psi_x\rangle \langle \psi_x|\}_{x \in \mathbb X})$, we give a necessary condition for a set of states to violate an inequality: the corresponding Gram matrix must be distinct from those generated in the scenario $S$. We have the following theorem, whose proof is written in Appendix \ref{app:proof_thm1}.
\begin{thm}
    Consider a prepare-and-measure witness $W$ with classical and stabilizer value $W^C$ and $W^\mathrm{STAB}$ respectively. Consider some set of states $\{\ket{\psi_x}\}_{x \in \mathbb{ X}} \subseteq \mathbb C^d$, with the corresponding $|\mathbb X|\times |\mathbb X|$ Gram matrix $G_{x, x^\prime}\equiv \langle \psi_x|\psi_{x^\prime}\rangle$. Given $ S \in \{C,\mathrm{STAB}\}$,
    \begin{enumerate}
        \item Given any subset $\{|\psi^S_x\rangle\}_x \subseteq \mathrm{ext}S$, the corresponding Gram matrices $G^S_{x,x^\prime} = \langle \psi^S_x|\psi^S_{x^\prime}\rangle$ is severely constrained in the classical and stabilizer scenarios: $G^C_{x,x^\prime} \in \{0,1\}$, and $G^\mathrm{STAB}_{x,x^\prime} \in \{0,1/\sqrt d, 1\}$.
        \item if $W(\{|\psi_x\rangle \langle \psi_x|\}_x) > W^S$, then, there is no such subset such that $G_{x,x^\prime} = G^S_{x,x^\prime}$.
    \end{enumerate}
    \label{thm:witness_overlap_thm}
\end{thm}
The theorem shows that state overlaps can be used to efficiently determine whether a set of states respects a given inequality. Consider a witness $W$ with a scenario (classical or stabilizer) bound denoted by $W^{\mathrm{S}}$. If one can construct a set of states ${\ket{\psi_x}}_{x \in \mathbb{X}} \subseteq \mathbb{C}^d$ such that their pairwise overlaps match those defined by $\mathrm{ext}S$, then these states cannot exceed the scenario bound, i.e., $W \leq W^{\mathrm{S}}$.

However, in general, the absence of such matching does not necessarily imply a violation of the prepare-and-measure witness value $W^{S}$.  This illustrates that, in general, employing a witness offers a stronger form of certification than merely comparing state overlaps. This idea is also reflected in Ref.~\cite{wagner2024certifyingnonstabilizernessquantumprocessors}, where the authors use overlap inequalities rather than relying solely on the overlaps themselves. A key distinction between our approach and theirs is that prepare-and-measure witnesses do not require knowledge of all pairwise overlaps within the considered set of states; it suffices to observe the overall statistics. The drawback, however, is that our method assumes the states are $d$-dimensional, which can be challenging to guarantee in practice.

Finally, we note that  prepare-and-measure witnesses can be seen trivially as overlap inequalities on the overlaps $O_{xy}^b = Tr(\rho_x M^y_b)$
\begin{equation}
    W =\sum_{b,x,y} W^b_{x,y}p(b|x,y)=\sum_{b,x,y} W^b_{x,y}O_{xy}^b.
\end{equation}
Ref.~\cite{PhysRevA.109.032220} demonstrates that inequalities for various quantum resources can be derived by considering polytopes defined by overlaps, such as those between states, or between measurements and states, as in the example above. Although the authors do not explicitly construct prepare-and-measure witnesses within this framework, it appears plausible that such witnesses could also be obtained using their overlap-based formalism.
As a less trivial observation, we note that some of the prepare-and-measure witnesses studied in this work can be written in terms of state overlaps $r_{ij} = Tr(\rho_i\rho_j)$ between state $\rho_i$ and $\rho_j$. For example, considering the form~\eqref{eq.S3withr} of Ineq.~$S_3$, by expressing the norms in terms of the dot product in the Euclidean space and by noting that for qubits the dot product of their Bloch vectors is related to their overlap as
$\vec{r}_{x}\cdot \vec{r}_{y} = 2r_{xy}-1 $, we can find that $S_3$ can be written as 
\begin{align}
    &\sqrt{-1 + 2(r_{11}+r_{22}+r_{33}) + 4(r_{12}-r_{13}-r_{23}) } \nonumber\\&+ \sqrt{2(r_{11} + r_{22}) - 4r_{12}}\leq3,
\end{align}
which for pure states reduces to 
\begin{equation}
    \sqrt{5+ 4(r_{12}-r_{13}-r_{23}) } + 2\sqrt{1 - r_{12}}\leq3.
\end{equation}

Similarly, one can show that the $T_N$ witnesses take the form 
\begin{equation}
    T_N = \frac{1}{\sqrt{2}} \sum_{i=1}^{N}\sqrt{\sum_{xy}(-1)^{x_i+y_i}r_{xy}}.
\end{equation}
We note that when expressed through states overlaps, the resulting inequalities cannot represent a polytope facet since they are not linear. It may be worth to contrast with Ref.~\cite{wagner2024certifyingnonstabilizernessquantumprocessors} where they use linear overlap inequalities for non-stabilizerness detection. As another note,  the inequalities they used were designed for detecting coherence, which is compatible with our results, reinforcing the idea that coherence and non-stabilizerness SDI detection may be treated within the same formalism. Regarding this observation and as a final note, we realized that in Ref.~\cite{PhysRevA.111.L040402} they constructed a witness to detect coherence, which ends up being able to detect non-stabilizerness as well. We found that the inequality they constructed is nothing else than $S_3$, which explains their findings. 
\\

\section{FINAL REMARKS}
\label{sec:FINAL REMARKS}

This work introduces a semi-device-independent framework for certifying non-stabilizer states in prepare-and-measure scenarios. By relying solely on a minimal assumption about the dimension of the quantum system, the proposed prepare-and-measure witnesses offer a robust and experimentally feasible approach to detecting non-stabilizerness.

In the simplest setting—three state preparations and two measurement settings—we identify concrete thresholds whose violation certifies the presence of non-stabilizer states. Notably, the maximal violation occurs uniquely for states in the Clifford orbit of the “H-state”, providing a kind self-testing method for the experimental verification of such states. Interestingly, we also observe that in certain cases, states with intermediate levels of non-stabilizerness yield the strongest witness violations, suggesting a nuanced relationship between this resource and operational detectability.

Extensions to scenarios with more preparations reveal further structure. In particular, for eight preparations, the observed witness values allow one to infer a lower bound on the number of non-stabilizer preparations involved. This quantitative feature supports more refined experimental protocols and informs the design of tasks such as quantum random access codes (QRACs), highlighting the operational relevance of non-stabilizer states.

We also propose a criterion based on state overlaps (Gram matrices), offering a necessary condition for non-stabilizerness certification. Specifically, the appearance of matrix elements outside the set $\{0, \frac{1}{\sqrt d}, 1\}$ is shown to be essential in $d$-dimensional systems. This connection points toward a deeper interplay between resource-theoretic concepts and semi-device-independent methods.

Altogether, our results provide a set of analytically grounded and practically applicable tools for the study and certification of non-stabilizer states. They contribute to the broader effort to understand and utilize the resources necessary for achieving quantum computational advantages within a minimal-assumption framework.

\section*{Acknowledgements}
We thank an anonymous referee for their insightful suggestion, which led to the results presented in Fig. 4 of our manuscript. This work was supported by the Simons Foundation (Grant Number 1023171, RC), the Brazilian National Council for Scientific and Technological Development (CNPq, Grants No.307295/2020-6, No.403181/2024-0 and 301687/2025-0), the Financiadora de Estudos e Projetos (grant 1699/24 IIF-FINEP) and the Coordenação de Aperfeiçoamento de Pessoal de Nível Superior – Brasil (CAPES) – Finance Code 001. This work has also been partially funded by the project "Comparative Analysis of P\&M Protocols for Quantum Cryptography" supported by QuIIN - Quantum Industrial Innovation, EMBRAPII CIMATEC Competence Center in Quantum Technologies, with financial resources from the PPI IoT/Manufatura 4.0 of the MCTI grant number 053/2023, signed with EMBRAPII. DP acknowledges funding from MUR PRIN (Project 2022SW3RPY).\\
\textbf{ Code availability.} The code used to generate the results in
this paper is available on GitLab: \href{https://gitlab.com/data_availability/sdi-nonstabilizerness-certification}{https://gitlab.com/data\_availability/sdi-nonstabilizerness-certification}

\bibliographystyle{apsrev4-2}
\bibliography{references}
\appendix
\section{  Stabilizer strategies  maximizing the $T_N$  witnesses for $N\in\{3,4,5\}$}\label{app: stab_optimal_strategies}

This appendix provides explicit stabilizer state configurations—represented as Bloch vectors—that maximize the prepare-and-measure inequality \(T_N\) for scenarios with \(N = 3\), \(4\), and \(5\) preparations. These configurations are derived from numerical optimizations restricted to stabilizer states.\\

\textbf{Case N=3}\\

The optimal stabilizer strategy for \(T_3\) involves the following eight qubit states (Bloch vectors):  
  
\begin{align}
&\vec{r}_1 = [ 0,  0,  1] \quad \vec{r}_2 = [ 0,  1, 0], \quad \vec{r}_3 = [-1, 0,  0], \\
&\vec{r}_4 = [0,  1, 0], \quad \vec{r}_5 = [ 0, -1,  0], \quad \vec{r}_6 = [ 1,  0,  0], \\
&\vec{r}_7 =[0, -1, 0], \quad \vec{r}_8 = [ 0,  0, -1].
\end{align}
these vectors correspond to antipodal eigenstates of the Pauli operators \(\sigma_X\), \(\sigma_Y\), and \(\sigma_Z\).\\

\textbf{Case N=4} \\

For \(T_4\), the stabilizer-optimal strategy uses 16 states with Bloch vectors:

\begin{equation*}
\begin{aligned}
&\vec{r}_1 = [ 0,  1,  0], \quad \vec{r}_2 = [0,  0, -1], \quad \vec{r}_3 = [ 0,  0,  1],\\
&\vec{r}_4 = [ 0, -1, 0], \quad \vec{r}_5 = [ 0,  1, 0], \quad \vec{r}_6 = [ 0,  1, 0],\\
&\vec{r}_7 = [ 0,  0,  1], \quad \vec{r}_8 =[0,  0,  1], \quad \vec{r}_9 = [0,  0, -1], \\
&\vec{r}_{10} = [ 0,  0, -1], \quad \vec{r}_{11} = [ 0, -1,  0],\quad \vec{r}_{12} = [0, -1, 0], \\
&\vec{r}_{13} = [0,  1, 0], \quad \vec{r}_{14} = [0, 0, -1], \quad \vec{r}_{15} = [0,  0,  1], \\
&\vec{r}_{16} = [ 0, -1,  0].
\end{aligned}
\end{equation*}
\\

\textbf{Case N=5}\\

For \(T_5\), the stabilizer-optimal strategy employs 32 states with Bloch vectors:

\begin{equation*}
\begin{aligned}
&\vec{r}_1 = [0, 0, 1], \quad \vec{r}_2 = [0, 0, 1], \quad \vec{r}_3 = [0, 0,  1], \\
&\vec{r}_4 = [0,  1,  0] \quad \vec{r}_5 = [0, 0,  1], \quad \vec{r}_6 = [0,  0,  1], \\
&\vec{r}_7 = [0, 0,  1], \quad \vec{r}_8 = [0,  1, 0], \quad \vec{r}_9 = [0,  0,  1], \\
&\vec{r}_{10} = [0,  1,  0], \quad \vec{r}_{11} = [0,  1,  0], \quad \vec{r}_{12} = [0,  1, 0], \\
&\vec{r}_{13} = [0, -1,  0], \quad \vec{r}_{14} = [0,  0, -1], \quad \vec{r}_{15} = [0,  0, -1], \\
&\vec{r}_{16} = [0,  0, -1], \quad \vec{r}_{17} = [0,  0,  1], \quad \vec{r}_{18} = [0, 0,  1], \\
&\vec{r}_{19} = [0,  0,  1], \quad \vec{r}_{20} = [0,  1, 0], \quad \vec{r}_{21} = [0, -1,  0], \\
&\vec{r}_{22} = [0, -1, 0], \quad \vec{r}_{23} = [0, -1, 0], \quad \vec{r}_{24} = [0,  0, -1], \\
&\vec{r}_{25} = [0, -1,  0], \quad \vec{r}_{26} = [0,  0, -1], \quad \vec{r}_{27} = [0, 0, -1], \\
&\vec{r}_{28} = [0,  0, -1], \quad \vec{r}_{29} = [0, -1, 0], \quad \vec{r}_{30} = [0,  0, -1], \\
&\vec{r}_{31} = [0, 0, -1], \quad \vec{r}_{32} = [0, 0, -1].
\end{aligned}
\end{equation*}

 In all these states, we note the XOR property proven in the main text.

\section{Proof of Theorem \ref{thm:witness_overlap_thm}}
 \label{app:proof_thm1}

    (1) Our goal is to derive an explicit form for the Gram matrix elements $G^{S}_{x,x^\prime} = \langle \psi^S_x|\psi^S_{x^\prime}\rangle$ for $S\in~\{C, \mathrm{STAB}\}$.
    \begin{itemize}
        \item Case $S = C$. For this scenario the derivation is straight forward since $\mathrm{ext}C$ is an orthonormal basis of $\mathbb C^d$, the corresponding Gram matrix is $G^C_{x, x\prime} = 1$, if $|\psi^C_x\rangle = \ket{\psi^C_{x^\prime}}$, and zero otherwise.
        
        \item Case $S=\mathrm{STAB}$.  Here, the derivation involves subtler considerations. For $d=2$, the Gram matrix was already obtained in Ref.~\cite{10.5555/2638682.2638691}, but here we will generalize for dimension $d$. We consider two states $|\mathcal S_x\rangle, |\mathcal S_{x^\prime}\rangle \in\mathbb C^d$, stabilized by $\mathcal S_x = \langle D_{a_x,b_x}\rangle \subseteq \mathcal P_d$ and $\mathcal S_{x^\prime}  = \langle D_{a_{x^\prime}, b_{x^\prime}}\rangle \subseteq \mathcal P_d$, correspondingly. The non-trivial vectors $(a_x, b_x) \in \mathbb F_d^2$ and $(a_{x^\prime}, b_{x^\prime}) \in \mathbb F_d^2$ fix the generators of the groups $\mathcal S_x$ and $\mathcal S_{x^\prime}$. The corresponding projectors are then written as
        \begin{align}
            \rho_x &\equiv |\mathcal S_x\rangle \langle \mathcal S_x| = \frac{1}{d}\sum_{k=0}^{d-1}D_{a_x, b_x}^k\;,\\
           \rho_{x^\prime}&\equiv  |\mathcal S_{x^\prime}\rangle \langle \mathcal S_{x^\prime}| = \frac{1}{d}\sum_{k=0}^{d-1}D_{a_{x^\prime}, b_{x^\prime}}^k\;.
        \end{align}
        so  $\mathrm{Tr}(\rho_x \rho_{x^\prime}) =|\langle \mathcal S_x| \mathcal S_{x^\prime}\rangle|^2 $ can be written as
        \begin{align}
            \mathrm{Tr}(\rho_x \rho_{x^\prime}) & =\frac{1}{d^2} \sum_{k,k^\prime=0}^{d-1} \mathrm{Tr}(D^k_{a_x,b_x}D^{k^\prime}_{a_{x^\prime}, b_{x^\prime}})\nonumber\\
            &= \frac{1}{d^2} \sum_{k,k^\prime=0}^{d-1}\mathrm{Tr}(D_{ka_x,kb_x}D_{k^\prime a_{ x^\prime}, k^\prime b_{x^\prime}})\nonumber\\
            &= \frac{1}{d^2} \sum_{k,k^\prime=0}^{d-1}\mathrm{Tr}(D^\dagger_{-ka_x,-kb_x}D_{k^\prime a_{ x^\prime}, k^\prime b_{x^\prime}})\nonumber\\
            &= \frac{1}{d} \sum_{k,k^\prime=0}^{d-1}\delta_{k^\prime a_{x^\prime},-k a_x} \delta_{k^\prime b_{x^\prime},-k b_x}\;.\label{eq: overlaps_Th1Proof}
        \end{align}
        Where in the second and third line we used the group multiplication given in Eq.~(\ref{eq: WHgroup_operation}) which allows to calculate the product $D^k_{a_x,b_x}  = D_{ka_x,kb_x}$  and the inverse $D_{ka_x,kb_x}\cdot D_{-ka_x,-kb_x} = I$. From the third to the fourth line we used the fact that displacement operators are orthogonal to the Hilbert-Schmidt inner product, $\langle D_{a,b}, D_{c,d}\rangle_\mathrm{HS} \equiv \mathrm{Tr}(D^\dagger_{a,b} D_{c,d}) = d \delta_{a,c}\delta_{b,d}$.  The conditions implied by the deltas in this fourth line can be understood as a system of equations given by
        \begin{equation}
        \begin{pmatrix}
            a_{x^\prime} & a_x\\
            b_{x^\prime} & b_x
        \end{pmatrix}
        \begin{pmatrix}
            k^\prime\\
            k
        \end{pmatrix}
        =
        \begin{pmatrix}
            0\\
            0
        \end{pmatrix},
        \end{equation}
       parametrized by $(a_x, b_x, a_{x^\prime}, b_{x^\prime}) \in \mathbb F^4_d$. Of course, the lack of a solution for the system implies  $\mathrm{Tr}(\rho_x \rho_{x^\prime}) = 0$. We assume there exists at least one solution. Consider the determinant, which is exactly the symplectic inner product of $(a_x, b_x)$ and $(a_{x^\prime}, b_{x^\prime})$:
        \begin{equation}
            \Delta = a_{x^\prime}b_x-a_x b_{x^\prime}\;.
        \end{equation}
        Consider first the case where $\Delta \neq 0 \mod d$. Then, the only solution is the trivial one and $k=k^\prime = 0$. This corresponds exactly to the case where $\mathrm{Tr}(\rho_x \rho_{x^\prime})=1/d$. Otherwise, assume that indeed $a_{x^\prime} b_x = a_x b_{x^\prime}$. Since we assume that we have non-trivial vectors $(a_x,b_x) \neq (0,0)$ and $(a_{x^\prime}, b_{x^\prime}) \neq (0,0)$,  the coefficient matrix has rank 1. This means that the solution can be described by only one equation of the form  $c_1k + c_2 k^\prime= 0 \mod d$ with $c_1$ and $c_2$ non-vanishing coefficients.  Therefore, the kronecker deltas in Eq.~(\ref{eq: overlaps_Th1Proof}) are all equal to $1$ implying $\mathrm{Tr}(\rho_x \rho_{x^\prime})= 1$. In conclusion, we then have shown that the possible values of the $G$ matrix in the stabilizer scenario are
        \begin{equation}
            G_{x,x^\prime}^\mathrm{STAB} \in \{0,1/\sqrt d,1\}.
        \end{equation}
    \end{itemize}
    (2) Assume that there is a set $\{|\psi^S_x\rangle\}_{x \in \mathbb X} \subseteq \mathrm{ext} S$, such that 
    \begin{equation}
        G_{x, x^\prime} = \langle \psi_x^S|\psi_{x^\prime}^S\rangle\;, \label{eq: condition_Gram_matrix}
    \end{equation}
    then $W(\{\ket{\psi_x} \bra{\psi_x}\}_x) \leq W^S$. This can be seen by noting that Proposition \ref{prop:unitaries_overlaps} guarantees that the set $\{\ket{\psi_x}\}_{x \in \mathbb{ X}}$ is unitarily equivalent to $\{|\psi^S_x\rangle\}_{x \in \mathbb X}$. Then Proposition \ref{prop:unitary_symmetry} implies  that $W(\{\ket{\psi_x} \bra{\psi_x}\}_x) \leq ~W^\mathrm{S}$. Hence, the violation of the inequality means that there is no subset of states with the same Gram matrix.
    $\square$
\end{document}